\documentclass[10pt,journal,onecolumn]{IEEEtran}

\usepackage{algorithm}
\usepackage{algpseudocode}
\usepackage{times,amssymb,amsmath,amsfonts,nicefrac,float,accents,color,
graphicx,bbm,caption,subcaption,mathrsfs,amsthm,stmaryrd,soul}
\usepackage[noadjust]{cite}
\usepackage{tikz}
\usetikzlibrary{calc}
\usepackage[mathscr]{eucal}
\usepackage{arydshln}
\usepackage{booktabs}
\usepackage{makecell}

\newcommand{\tabincell}[2]{\begin{tabular}{@{}#1@{}}#2\end{tabular}}

\newtheorem{thm}{Theorem}[section]
\newtheorem{rem}[thm]{Remark}
\newtheorem{definition}[thm]{Definition}
\newtheorem{lem}[thm]{Lemma}

\newtheorem{cor}[thm]{Corollary}

\newtheorem{claim}[thm]{Claim}

\title{New Constructions of Optimal Locally Repairable Codes with Super-Linear Length}

\author{Xiangliang Kong$^{\text{a}}$, Xin Wang$^{\text{b}}$\thanks{The research of X. Wang was supported by the National Natural Science Foundation of China under Grant No. 11801392 and the Natural Science Foundation
of Jiangsu Province under Grant No. BK20180833.} and Gennian Ge$^{\text{a}}$\thanks{Corresponding author. Email address:  gnge@zju.edu.cn (G. Ge). The research of G. Ge was supported by the National Natural Science Foundation of China under Grant No. 11971325, National Key Research and Development Program of China under Grant  Nos. 2020YFA0712100  and  2018YFA0704703,  and Beijing Scholars Program.}\\
  \footnotesize $^{\text{a}}$ School of Mathematical Sciences, Capital Normal University, Beijing, 100048, China\\
  \footnotesize $^{\text{b}}$ School of Mathematical Sciences, Soochow University, Soochow 215301, Jiangsu, China}

\begin{document}
\date{}
\maketitle

\begin{abstract}
As an important coding scheme in modern distributed storage systems, locally repairable codes (LRCs) have attracted a lot of attentions from perspectives of both practical applications and theoretical research. As a major topic in the research of LRCs, bounds and constructions of the corresponding optimal codes are of particular concerns.

In this work, codes with $(r,\delta)$-locality which have optimal minimal distance w.r.t. the bound given by Prakash et al. \cite{Prakash2012Optimal} are considered. Through parity check matrix approach, constructions of both optimal $(r,\delta)$-LRCs with all symbol locality ($(r,\delta)_a$-LRCs) and optimal $(r,\delta)$-LRCs with information locality ($(r,\delta)_i$-LRCs) are provided. As a generalization of a work of Xing and Yuan \cite{XY19}, these constructions are built on a connection between sparse hypergraphs and optimal $(r,\delta)$-LRCs. With the help of constructions of large sparse hypergraphs, the length of codes constructed can be super-linear in the alphabet size. This improves upon previous constructions when the minimal distance of the code is at least $3\delta+1$. As two applications, optimal H-LRCs with super-linear length and GSD codes with unbounded length are also constructed.

\medskip
\noindent{\it Keywords:} Optimal locally repairable codes, parity-check matrix, sparse hypergraphs

\smallskip


\end{abstract}

\section{Introduction}

In modern distributed storage systems, erasure coding based schemes are employed to provide efficient repair for failed storage nodes. Among all these storage codes, maximum distance separable (MDS) codes are favored for their high repair efficiency and reliability. However, due to the large bandwidth and disk I/O during repair process (see \cite{MMDARSD13}), schemes based on MDS codes can be costly when only a few nodes fail in the system. This greatly affects the practicability of MDS codes in storage systems, especially in large-scale distributed file systems.

To maintain high repair efficiency with less bandwidth, locally repairable codes (LRCs) were introduced in \cite{gopalan2012locality}. A block code is called a locally repairable code with locality $r$ if any failed code symbol can be recovered by accessing at most $r$ survived ones. Moreover, if this code is linear, $r$ should be much smaller than code dimension $k$. Therefore, LRCs can guarantee efficient recovery of single node failures with low repair bandwidth. As a result, LRCs have been implemented in many large scale systems e.g., Microsoft Azure \cite{CHYABPJS12} and Hadoop HDFS \cite{MMDARSD13}.

Over the past few years, the concept of LRCs has been generalized in many different aspects. As one major generalization, the notion of locally repairable codes with $(r,\delta)$-locality ($(r,\delta)$-LRCs) was introduced by Prakash et al. \cite{Prakash2012Optimal}, which extends the capability of repairing one erasure within each repair set to $\delta-1$ erasures. Like original LRCs, a Singleton-type upper bound on the minimum distance of $(r,\delta)$-LRCs was given in \cite{Prakash2012Optimal}. Recently, finding constructions of the optimal LRCs and optimal $(r,\delta)$-LRCs with respect to such bounds has become an interesting and challenging work, which attracted lots of researchers. For examples, see \cite{JMX20,Jin19,LMX19,LXY19,SRKV13,TB14,TPD16,XY19} for constructions of optimal LRCs and see \cite{CMST20,CXHF18,ZL20,Zhang20} for constructions of optimal $(r,\delta)$-LRCs. In this paper, we focus on constructions of optimal $(r,\delta)$-LRCs. For the study of availabilities of LRCs, see \cite{CCFT19,CMST20,RPDV16,SES19,TB16,WZ14}, and for the study of codes with hierarchical locality (H-LRCs), see \cite{ZL20,SAK15,BBV2019,LC20}. For other generalizations, we refer to the survey \cite{balaji2018erasure}.

Usually, longer codes over smaller fields are favored for their efficient transmission performances and fast implementations in practical applications. Therefore, given the size $q$ of the underlying field and other parameters, it is natural to ask how long a code with such parameters can be. For optimal $(r,\delta)$-LRCs, this question was recently asked by Guruswami et al. \cite{guruswami2019long}. They considered this question for the case $\delta=2$ and proved an upper bound on the code length. Through a greedy algorithm, they also constructed optimal $(r,2)$-LRCs with super-linear (in $q$) length, which confirmed the tightness of their upper bound for some cases. Latter in \cite{CMST20}, Cai et al. considered this question for the general case $\delta>2$. In their paper, Cai et al. derived a general upper bound on the length of optimal $(r,\delta)$-LRCs and they also provided a general construction for such codes. Furthermore, using combinatorial objects such as union-intersection-bounded families, packings and Steiner systems, they obtained optimal $(r,\delta)$-LRCs with length $\Omega(q^{\delta})$, which meet their upper bound on the code length when the minimal distance $d$ satisfies $2\delta+1\leq d\leq 3\delta$. Very recently, Cai and Schwartz \cite{CS20} extended their results in \cite{CMST20} to codes that not only have information $(r,\delta)$-locality but also can recover some erasure patterns beyond the minimum distance. They also introduced a new kind of array codes called \emph{generalized sector-disk} (GSD) codes, which can recover special erasure patterns mixed of whole disk erasures together with additional sector erasures that are beyond the minimum distance.


In this paper, through parity-check matrix approach, we provide general constructions for both optimal $(r,\delta)$-LRCs with all symbol locality and optimal $(r,\delta)$-LRCs with information locality and extra global recoverability. Our constructions are built on a connection between sparse hypergraphs in extremal combinatorics and optimal $(r,\delta)$-LRCs, which can be viewed as a generalization of a work of Xing and Yuan \cite{XY19}. Based on known results and a probabilistic construction about sparse hypergraphs, we obtain optimal $(r,\delta)_a$-LRCs (codes with all symbol $(r,\delta)$-locality) and optimal $(r,\delta)_i$-LRCs (codes with information $(r,\delta)$-locality) with length super-linear in $q$. Compared to the results in \cite{CMST20} and \cite{CS20}, our results provide longer codes for $d\geq 3\delta+1$. Furthermore, as two applications of our constructions, we construct optimal H-LRCs with super-linear length, which improves the results given by \cite{Zhang20}; and we also provide a construction of generalized sector-disk codes with unbounded length.

The remainder of this paper is organized as follows. In Section \ref{sec_pre}, we fix some notations and provide preliminaries on locally repairable codes. In Section \ref{sec3}, we present our constructions of optimal $(r,\delta)$-LRCs with all symbol locality. In Section \ref{sec4}, we present our constructions of optimal $(r,\delta)$-LRCs with information locality and extra global recoverability. In Section \ref{sec5}, we first give a brief introduction about Tu\'{r}an-type problems for sparse hypergraphs, and then based on constructions of a special kind of sparse hypergraphs, we obtain optimal $(r,\delta)_a$-LRCs and optimal $(r,\delta)_i$-LRCs with super-linear length. In Section \ref{sec6}, we provide two applications of our constructions for H-LRCs and GSD codes. Finally, we conclude our paper with some remarks in Section \ref{sec7}.

\section{Preliminaries}\label{sec_pre}

Firstly, we introduce some notations and terminologies that will be frequently used throughout the paper:
\begin{itemize}
  \item[1)] Let $q$ be a prime power, we define $\mathbb{F}_q$ as the finite field with $q$ elements. $\mathcal{C}$ is said to be an $[n,k,d]_q$ code if $\mathcal{C}$ is a linear code over $\mathbb{F}_q$ with length $n$, dimension $k$ and minimum distance $d$.
  \item[2)] For positive integer $n$, we use $[n]=\{1,2,\cdots,n\}$ as the first $n$ positive integers. For $x\geq 0$, we use $\lfloor x\rfloor$ and $\lceil x\rceil$ to denote the floor function and ceiling function of $x$, respectively.
  \item[3)] We use $O$ to denote the zero matrix with proper size according to the context.
  \item[4)] For a vector $\mathbf{v}\in \mathbb{F}_q^{n}$, define $\omega(\mathbf{v})\triangleq|\{i\in[n]: v(i)\neq 0\}|$ as the weight of $\mathbf{v}$.
  \item[5)] A vector over $\mathbb{F}_q$ is said to be Vandemonde-type with generator (or generating element) $a$ if it has the form $b(1,a,a^2,\cdots)^T$ for some $a,b\in \mathbb{F}_{q}^{*}$. For a set of Vandemonde-type vectors $\{\mathbf{v}_i\}_{i=1}^{s}$, the generating set of $\{\mathbf{v}_i\}_{i=1}^{s}$ consists of all the generators for every $\mathbf{v}_i$, $1\leq i\leq s$.
  \item[6)] For positive integers $m$ and $n$, let $E$ be a subset of $[n]$ with size $s$. Write $E=\{i_1,\ldots,i_s\}$ when $s\geq 1$ and $E=\emptyset$ when $s=0$. Let $\mathbf{H}=(\mathbf{h}_1,\mathbf{h}_2,\ldots,\mathbf{h}_n)$ be a matrix of size $m\times n$, where $\mathbf{h}_i\in \mathbb{F}_q^{m}$ for $1\leq i\leq n$. Then, the restriction of $\mathbf{H}$ over $E$ is defined as $\mathbf{H}|_{E}=(\mathbf{h}_{i_1},\mathbf{h}_{i_2},\ldots,\mathbf{h}_{i_s})$ when $s\geq 1$ and $\mathbf{H}|_{E}=()$, i.e., the empty matrix, when $s=0$.
  \item[7)] We use the standard Bachmann-Landau notations $\Omega(\cdot)$, $\theta(\cdot)$, $O(\cdot)$ and $o(\cdot)$, whenever the constant factors are not important.
\end{itemize}

Now we state the formal definition of $(r,\delta)$-locality.
\begin{definition}\emph{(\cite{Prakash2012Optimal})}\label{def01} Let $\mathcal{C}$ be an $[n,k,d]_q$ code. The $i$th code symbol $c_i$ of $\mathcal{C}$ is called to have locality $(r,\delta)$ if there exists a subset $S_i\subset[n]$ satisfying
\begin{itemize}
\item[$\bullet$] $i\in S_i$ and $|S_i|\leq r+\delta-1$,
\item[$\bullet$] the minimum distance of the code $\mathcal{C}|_{S_i}$ obtained by deleting code symbols $c_i$, $i\in[n]\backslash S_i$, is at least $\delta$.
\end{itemize}
\end{definition}
An $[n,k,d]_q$ code $\mathcal{C}$ is said to have all symbol $(r,\delta)$-locality ($(r,\delta)_a$-locality) if all symbols of $\mathcal{C}$ have locality $(r,\delta)$ and it is said to have information $(r,\delta)$-locality ($(r,\delta)_i$-locality), if there exists a $k$-set $I\subseteq [n]$ with $rank(I)=k$ such that for every $i\in I$, the $i_{th}$ symbol has $(r,\delta)$-locality. As shown in \cite{Prakash2012Optimal}, for both $[n,k,d]_q$ codes with $(r,\delta)_a$-locality and $[n,k,d]_q$ codes with $(r,\delta)_i$-locality, their minimal distance $d$ satisfies the following Singleton-type bound:
\begin{equation}\label{singleton_bound}
d\leq n-k+1-(\left\lceil\frac{k}{r}\right\rceil-1)(\delta-1).
\end{equation}
When the equality in (\ref{singleton_bound}) holds, the code $\mathcal{C}$ is called optimal. For the sake of our construction, we change the form of the Singleton-type bound as follows.
\begin{lem}\label{lem01} Assume that $(r+\delta-1)|n$. If the Singleton-type bound \emph{(\ref{singleton_bound})} is achieved, then
\begin{equation}n-k=(\delta-1)\frac{n}{r+\delta-1}+d-\delta-(\delta-1)\left\lfloor{\frac{d-\delta}{r+\delta-1}}\right\rfloor.
\end{equation}
\end{lem}
\begin{proof} Suppose that $d=n-k+1-(\lceil\frac{k}{r}\rceil-1)(\delta-1)$. Write $k=ar-b$ for some integers $a\geq1$ and $0\leq b\leq r-1$. Substituting $k=ar-b$ back into (\ref{singleton_bound}), we can get $d=n-(ar-b)-(\delta-1)a+\delta=n-(r+\delta-1)a+b+\delta$.
This implies that $a=\frac{n}{r+\delta-1}-\frac{d-\delta-b}{r+\delta-1}$. Since $\frac{n}{r+\delta-1}$ is an integer, thus $(r+\delta-1)|d-b-\delta$. Therefore, we further have $a=\frac{n}{r+\delta-1}-\lfloor{\frac{d-\delta}{r+\delta-1}}\rfloor$.
Finally, the result follows from
$$d=n-k-(a-1)(\delta-1)+1
=n-k-\frac{n}{r+\delta-1}(\delta-1)-\left\lfloor{\frac{d-\delta}{r+\delta-1}}\right\rfloor(\delta-1)+\delta. \qedhere $$
\end{proof}
\begin{rem} Similar results are shown in \cite{guruswami2019long} and \cite{XY19} for the case $\delta=2$.
\end{rem}
\begin{rem} When $r=d-\delta$ and $(r+\delta-1)|n$, the Singleton-type bound \emph{(\ref{singleton_bound})} can't be met. Indeed, let $x$ be the least nonnegative integer satisfying $$d+x=n-k+1-(\left\lceil\frac{k}{r}\right\rceil-1)(\delta-1).$$
Since we can assume $k=ar-b$ for some integers $a\geq1$ and $0\leq b\leq r-1$, thus we have $$d+x=n-a(r+\delta-1)+b+\delta.$$
As $r=d-\delta$, it follows that $$r+x-b=n-a(r+\delta-1).$$
Then $r+x-b$ must be divisible by $r+\delta-1$. So $x=\delta-1+b$. This indicates that the minimum distance $d$ of $\mathcal{C}$ is upper bounded by \begin{equation}\label{singleton_bound2}
d\leq n-k+1-\left\lceil\frac{k}{r}\right\rceil(\delta-1).
\end{equation}
When $\mathcal{C}$ meets the bound (\ref{singleton_bound2}), we say it is optimal for this case. When $\delta=2$, such phenomenon has already appeared in \cite{gopalan2012locality}.
\end{rem}

\section{Constructions of $(r,\delta)$-LRCs with all symbol locality}\label{sec3}

In this section, we consider linear codes with all symbol $(r,\delta)$-locality and provide a general construction for optimal $[n,k,d; (r,\delta)_a]_q$-LRCs through parity-check matrix approach. Compared to the constructions in \cite{ZL20} and \cite{Zhang20}, the restrictions of the parity-check matrix in our construction are more relaxed and therefore, our construction will lead to longer codes.

\subsection{Construction A}\label{ssec31}

Let $d\geq \delta+1$, $R=r+\delta-1$ and $n=mR$. For $i\in[m]$, let $G_i=\{g_{i,1},g_{i,2},\cdots,g_{i,R}\}$ be an $R$-subset of $\mathbb{F}_q$. Then, for each $i\in [m]$, we can construct a $(d-1)\times R$ Vandermonde matrix with generating set $G_i$ of the form
$\left(\begin{array}{c}
\mathbf{U}_i\\
\mathbf{V}_i
\end{array}\right)$, where
\begin{equation*}
\mathbf{U}_i= \left(\begin{array}{cccc}
1 &1 &\cdots &1\\
g_{i,1} &g_{i,2} &\cdots &g_{i,R}\\
\vdots &\vdots &\ddots &\vdots\\
g^{\delta-2}_{i,1} &g^{\delta-2}_{i,2} &\cdots &g^{\delta-2}_{i,R}
\end{array} \right)\text{~and~}
\mathbf{V}_i= \left(\begin{array}{cccc}
g^{\delta-1}_{i,1} &g^{\delta-1}_{i,2} &\cdots &g^{\delta-1}_{i,R}\\
g^{\delta}_{i,1} &g^{\delta}_{i,2} &\cdots &g^{\delta}_{i,R}\\
\vdots &\vdots &\ddots &\vdots\\
g^{d-2}_{i,1} &g^{d-2}_{i,2} &\cdots &g^{d-2}_{i,R}
\end{array} \right).
\end{equation*}
Note that $\mathbf{U}_i$ is a Vandermonde matrix of size $(\delta-1)\times R$ and $\mathbf{V}_i$ is of size $(d-\delta)\times R$. Put
\begin{equation}\label{pcm1}
\mathbf{H}=\left(\begin{array}{cccc}
\mathbf{U}_1 &O &\cdots &O\\
O &\mathbf{U}_2 &\cdots &O\\
\vdots &\vdots &\ddots &\vdots\\
O &O &\cdots &\mathbf{U}_m\\
\mathbf{V}_1 &\mathbf{V}_2 &\cdots &\mathbf{V}_m
\end{array}\right).
\end{equation}
Let $\mathcal{C}$ be the linear code with parity-check matrix $\mathbf{H}$. Due to the structure of $\mathbf{H}$ and the property of Vandermonde matrix $\mathbf{U}_i$, it is immediate from Definition \ref{def01} that $\mathcal{C}$ has all symbol $(r,\delta)$-locality. On the other hand, $\mathcal{C}$ has dimension
\begin{equation*}
k(\mathcal{C})\geq n-(\delta-1)m-(d-\delta)=rm-(d-\delta).
\end{equation*}
When $r>d-\delta$, we have $\lceil\frac{k(\mathcal{C})}{r}\rceil\geq m$. Thus, $\mathcal{C}$ has minimum distance
\begin{equation*}
d(\mathcal{C})\leq n-k(\mathcal{C})+\delta-\lceil\frac{k(\mathcal{C})}{r}\rceil(\delta-1)\leq d.
\end{equation*}
As for $r=d-\delta$, there is still $d(\mathcal{C})\leq d$ with respect to (\ref{singleton_bound2}). Therefore, for $r\geq d-\delta$, in order to obtain an optimal $[n,k,d;(r,\delta)_a]$-LRC from the above construction, it suffices to show that the minimum distance of $\mathcal{C}$ equals $d$. More precisely, our following aim is to find $m$ $R$-subsets $G_1,G_2,\cdots,G_m$ in $\mathbb{F}_q$ such that any $d-1$ columns from the matrix $\mathbf{H}$ are linearly independent. For brevity, we refer to the $i$th block as the set of columns from $\mathbf{H}$ where $\mathbf{U}_i$ arises. So there are $m$ column blocks and each one is made up of $R$ columns.

We finish this subsection with the following two observations about columns in $\mathbf{H}$:
\begin{description}
  \item[\textbf{Obs.1}]: any $d-1$ columns in a single block are linearly independent;
  \item[\textbf{Obs.2}]: any $\delta-1$ columns from one block are linearly independent from columns belonging to other blocks.
\end{description}

\subsection{Optimal LRCs with $(r,\delta)_a$-locality from Construction A}\label{ssec32}

In this subsection, we put some sufficient conditions on the generating sets $G_1, G_2, \cdots, G_m$ in Construction A to guarantee the optimality of minimum distances w.r.t. bounds (\ref{singleton_bound}) and  (\ref{singleton_bound2}). As a warm up, we start with the construction of optimal $(r,\delta)_a$-LRCs with small minimum distance and unbounded length. It is worth noting that Zhang and Liu also proved the following result in \cite{ZL20}, for the completeness of this paper, we include the result here.
\begin{thm}\cite{ZL20}\label{thm31}
Let $\delta+1\leq d\leq2\delta$, set $G_1=G_2=\cdots=G_m$, then any $d-1$ columns of $\mathbf{H}$ are linearly independent.
\end{thm}
\begin{proof}
Pick any $d-1$ columns from $H$. To see whether these columns are linearly independent, it suffices to consider the case where only one block contains at least $\delta$ columns because of $\delta+1\leq d\leq2\delta$ and \textbf{Obs.2}. Then combining \textbf{Obs.1} with \textbf{Obs.2}, we can conclude that these $d-1$ columns of $\mathbf{H}$ are linearly independent.
\end{proof}
\begin{cor}\label{cor31}
Let $q\geq r+\delta-1$. Assume that $\delta+1\leq d\leq2\delta$, $r\geq d-\delta$ and $(r+\delta-1)|n$, then there exist optimal $[n,k,d;(r,\delta)_a]$-LRCs with $n=m(r+\delta-1)$ for any positive integer $m$.
\end{cor}

For $d\geq 2\delta+1$, we have the following theorem.
\begin{thm}\label{main0}
Let $d\geq 2\delta+1$ and $r\geq d-\delta$. Suppose that for any subset $S\subseteq[m]$ with $2\leq |S|\leq\lfloor\frac{d-1}{\delta}\rfloor$, we have \begin{equation}\label{main_cond}
|\bigcup_{i\in S}G_i|\geq (r+\frac{\delta}{2}-1)|S|+\frac{\delta}{2},
\end{equation}
then any $d-1$ columns of $\mathbf{H}$ are linearly independent. As a result, the code $\mathcal{C}$ generated by Construction A in Section \ref{ssec31} is an optimal $[n,k,d;(r,\delta)_a]_q$-LRC.
\end{thm}

\begin{proof}
For $1\leq i\leq m$ and $1\leq j\leq R$, let $\mathbf{h}_{i,j}$ be the $j_{th}$ column from the $i_{th}$ block of $\mathbf{H}$, i.e.,
\begin{equation*}
\begin{array}{c}
\mathbf{h}_{i,j}=\begin{array}{c}
(0,0,\ldots,0,1,g_{i,j},\ldots,g_{i,j}^{\delta-2},0,0,\ldots,0,g_{i,j}^{\delta-1},g_{i,j}^{\delta},\ldots,g_{i,j}^{d-2})^{T}.
\end{array}
\\[-8pt]
 \begin{array}{cllllc}
 &~\underbrace{\rule{14mm}{0mm}}_{(i-1)(\delta-1)} & ~~~~~~~~~~~~~~~~ & \underbrace{\rule{14mm}{0mm}}_{(m-i)(\delta-1)} & ~~~~~~~~~~~~~~~~ &
 \end{array}
 \end{array}
\end{equation*}
Assume that there exist $d-1$ columns $\{\mathbf{h}_{i_1,j_1},\mathbf{h}_{i_2,j_2},\ldots,\mathbf{h}_{i_{d-1},j_{d-1}}\}$ in $\mathbf{H}$ that are linearly dependent. Then, we have
\begin{equation}\label{eq001}
\sum_{l=1}^{d-1}\lambda_{l}\mathbf{h}_{i_l,j_l}=\mathbf{0}.
\end{equation}
For $1\leq i\leq m$, denote $E_{i}=\{j_l: \lambda_l\neq 0~\text{and}~i_l=i \}$. Clearly, we have $\sum_{i\in [m]}|E_i|\leq d-1$. According to the structure of $\mathbf{H}$, we know that either $|E_i|=0$ or $|E_i|\geq \delta$. Otherwise, one shall get $|E_i|$ distinct columns linearly dependent in $\mathbf{U}_i$, which contradicts to the property of $\mathbf{U}_i$.

W.l.o.g., assume that $\{i:E_i\neq \emptyset\}=[t]$ and for each $i\in [t]$, $|E_i|=s_i$. Clearly, we have $t\leq\frac{d-1}{\delta}$. For each $i\in [t]$, denote $F_{i}=\{g_{i,j}\in G_{i}: j\in E_i\}$ as the generating set of columns corresponding to $E_i$ and $F=\bigcup_{i\in [t]}F_i$. Denote $\mathbf{H}'$ as the $(m(\delta-1)+d-\delta)\times (\sum_{i\in [t]}s_i)$ submatrix of $\mathbf{H}$ consisting of columns indexed by $\bigcup_{i\in [t]}\{(i,j): j\in E_i\}$. Write $F_i=\{a_{i,1},\ldots,a_{i,s_i}\}$, then,
$\mathbf{H}'$ has the following form:
\begin{equation}\label{H'0}
\mathbf{H}'=\left(\begin{array}{cccc}
\mathbf{A}_1 &O &\cdots &O\\
O &\mathbf{A}_2 &\cdots &O\\
\vdots &\vdots &\ddots &\vdots\\
O &O &\cdots & \mathbf{A}_t\\
O &O &\cdots & O\\
\vdots &\vdots &\ddots &\vdots\\
\mathbf{B}_1 & \mathbf{B}_2 & \cdots & \mathbf{B}_t
\end{array}\right),
\end{equation}
where
\begin{equation*}
\mathbf{A}_i=\left(\begin{array}{cccc}
1 & 1 & \cdots & 1 \\
a_{i,1} & a_{i,2} & \cdots & a_{i,s_{i}} \\
\vdots & \vdots & \ddots & \vdots \\
a_{i,1}^{\delta-2} & a_{i,2}^{\delta-2} & \cdots & a_{i,s_{i}}^{\delta-2}
\end{array}\right)
\text{~and~}
\mathbf{B}_i=\left(\begin{array}{cccc}
a_{i,1}^{\delta-1} & a_{i,2}^{\delta-1} & \cdots & a_{i,s_{i}}^{\delta-1} \\
a_{i,1}^{\delta} & a_{i,2}^{\delta} & \cdots & a_{i,s_{i}}^{\delta} \\
\vdots & \vdots & \ddots & \vdots \\
a_{i,1}^{d-2} & a_{i,2}^{d-2} & \cdots & a_{i,s_{i}}^{d-2}
\end{array}\right).
\end{equation*}

Denote $F_1^{1}=F_1$ and $F_i^{1}=F_{i}\setminus \bigcup_{j=1}^{i-1}F_j$ for $2\leq i\leq t$. Then, we have $F=\sqcup_{i=1}^{t}F_{i}^{1}$. By permutating the columns of $\mathbf{H}'$, we can obtain a matrix of the following form:
\begin{equation*}\label{H_20}
\mathbf{H}_2=(\mathbf{H}_L ||~\mathbf{H}_R)=\left(\begin{array}{ccccc||cccc}
\mathbf{A}_1& O & O & \cdots & O & O & O & \cdots & O\\
O & \mathbf{A}_2^{1} & O & \cdots & O & \mathbf{A}_2^{2} & O & \cdots & O\\
O & O & \mathbf{A}_3^{1} & \cdots & O & O & \mathbf{A}_3^{2} & \cdots & O\\
\vdots & \vdots & \vdots & \ddots & \vdots & \vdots & \vdots & \ddots & \vdots\\
O & O & O & \cdots & \mathbf{A}_t^{1} & O & O & \cdots & \mathbf{A}_t^{2}\\
O & O & O & \cdots & O & O & O & \cdots & O \\
\vdots & \vdots & \vdots & \ddots & \vdots & \vdots & \vdots & \ddots & \vdots \\
\mathbf{B}_1 & \mathbf{B}_2^{1} & \mathbf{B}_3^{1} & \cdots & \mathbf{B}_t^{1} & \mathbf{B}_2^{2} & \mathbf{B}_3^{2} & \cdots & \mathbf{B}_t^{2}
\end{array}\right),
\end{equation*}
where for $2\leq i\leq t$, $\mathbf{A}_i^{1}=\mathbf{A}_i|_{F_{i}^1}$, $\mathbf{A}_i^{2}=\mathbf{A}_i|_{F_{i}\setminus F_{i}^1}$ and $\mathbf{B}_i^{1}=\mathbf{B}_i|_{F_{i}^1}$, $\mathbf{B}_i^{2}=\mathbf{B}_i|_{F_{i}\setminus F_{i}^1}$. \footnote{Given a Vandermonde matrix $\mathbf{A}$ with generating set $G$, for simplicity, we denote $\mathbf{A}|_{F}$ as the matrix obtained by restricting $\mathbf{A}$ to the columns corresponding to those elements in $F\subseteq G$.} Similarly, denote $E_{i}^{1}=\{j\in E_i: g_{i,j}\in F_{i}^{1}\}$ as the index set of columns generated by $F_{i}^{1}$. Then, (\ref{eq001}) can be written as
\begin{equation}\label{eq002}
\sum_{i=1}^{t}\sum_{j\in E_i^{1}}\lambda_{i,j}\mathbf{h}_{i,j}+\sum_{i=1}^{t}\sum_{j\in E_i\setminus E_{i}^{1}}\lambda_{i,j}\mathbf{h}_{i,j}=\mathbf{0},
\end{equation}
where $\lambda_{i,j}\neq 0$ is the relabeled $\lambda_l$ for $(i,j)=(i_l,j_l)$.

Note that for $2\leq i\leq t$ and each column in $\mathbf{B}_{i}^{2}$, its generating element in $F$ has already appeared in $F_{i'}^{1}$ for some $1\leq i'<i$. Therefore, we can do the following elementary row and column operations on $\mathbf{H}_2$:
\begin{itemize}
  \item First, for each $2\leq i\leq t$ and each column $\mathbf{h}_{i,j}$ in $(O~\cdots~O~(\mathbf{A}_i^{2})^{T}~O~\cdots~(\mathbf{B}_{i}^{2})^{T})^{T}$ of $\mathbf{H}_R$, subtract the column $\mathbf{h}_{i',j'}$ in $(O~\cdots~O~(\mathbf{A}_{i'}^{1})^{T}~O~\cdots~(\mathbf{B}_{i'}^{1})^{T})^{T}$ of $\mathbf{H}_L$ from it, where $(i',j')$ satisfies $i'<i$, $j'\in E_{i'}$ and $g_{i,j}=g_{i',j'}\in F_{i'}^{1}$. This leads to a matrix equivalent to $\mathbf{H}_2$:
      \begin{equation*}\label{H_2'0}
      \mathbf{H}_2'=(\mathbf{H}_L ||~\mathbf{H}_R')=\left(\begin{array}{cccc||cccc}
      \mathbf{A}_1& O & \cdots & O & -\mathbf{A}_{2,1}^{2} & -\mathbf{A}_{3,1}^{2} & \cdots & -\mathbf{A}_{t,1}^{2}\\
      O & \mathbf{A}_2^{1} & \cdots & O & \mathbf{A}_2^{2} & -\mathbf{A}_{3,2}^{2} & \cdots & -\mathbf{A}_{t,2}^{2}\\
      O & O & \cdots & O & O & \mathbf{A}_3^{2} & \cdots & -\mathbf{A}_{t,3}^{2}\\
      \vdots & \vdots & \ddots & \vdots & \vdots & \vdots & \ddots & \vdots\\
      O & O & \cdots & \mathbf{A}_t^{1} & O & O & \cdots & \mathbf{A}_t^{2}\\
      O & O & \cdots & O & O & O & \cdots & O \\
      \vdots & \vdots & \ddots & \vdots & \vdots & \vdots & \ddots & \vdots \\
      \mathbf{B}_1 & \mathbf{B}_2^{1} & \cdots & \mathbf{B}_t^{1} & O & O & \cdots & O
      \end{array}\right),
      \end{equation*}
      where for $1\leq j<i\leq t$, the $l_{th}$ column of $\mathbf{A}_{i,j}^{2}$ is identical to the $l_{th}$ column of $\mathbf{A}_i^{2}$ if the corresponding generating element appears in $F_{j}^{1}\cap F_{i}$ and is identical to the zero vector, otherwise. Clearly, we have $\mathbf{A}_{2,1}^{2}=\mathbf{A}_2^{2}$ and $\sum_{j=1}^{i-1}\mathbf{A}_{i,j}^{2}=\mathbf{A}_i^{2}$. Moreover, (\ref{eq002}) turns into:
      \begin{equation}\label{eq003}
      \sum_{i=1}^{t}\sum_{j\in E_i^{1}}\lambda'_{i,j}\mathbf{h}_{i,j}+\sum_{i=1}^{t}\sum_{j\in E_i\setminus E_{i}^{1}}\lambda_{i,j}\mathbf{h}'_{i,j}=\mathbf{0},
      \end{equation}
      where for $(i,j)\in [t]\times E_{i}^{1}$,
      \begin{equation*}
      \lambda_{i,j}'=\lambda_{i,j}+\sum_{i'>i}\sum_{\substack{j'\in E_{i'}:\\ g_{i',j'}=g_{i,j}}}\lambda_{i',j'};
      \end{equation*}
      and for $(i,j)\in [t]\times E_{i}\setminus E_{i}^{1}$,
      \begin{equation*}
      \mathbf{h}_{i,j}'=\mathbf{h}_{i,j}-\mathbf{h}_{i',j'}\in \mathbf{H}_{R}'
      \end{equation*}
      for some $(i',j')$ satisfying $i'<i$, $j'\in E_{i'}$ and $g_{i,j}=g_{i',j'}\in F_{i'}^{1}$.
  \item Second, for each $1\leq i\leq \delta-1$, add the $(i+j(\delta-1))_{th}$ row to the $i_{th}$ row for all $1\leq j\leq t-1$. Then, we have:
      \begin{equation}\label{H_2''0}
      \mathbf{H}_2''=(\mathbf{H}_L'' ||~\mathbf{H}_R'')=\left(\begin{array}{cccc||cccc}
      \mathbf{A}_{1} & \mathbf{A}_{2}^{1} & \cdots & \mathbf{A}_{t}^{1} & O & O & \cdots & O\\
      O & \mathbf{A}_2^{1} & \cdots & O & \mathbf{A}_2^{2} & -\mathbf{A}_{3,2}^{2} & \cdots & -\mathbf{A}_{t,2}^{2}\\
      O & O & \cdots & O & O & \mathbf{A}_3^{2} & \cdots & -\mathbf{A}_{t,3}^{2}\\
      \vdots & \vdots & \ddots & \vdots & \vdots & \vdots & \ddots & \vdots\\
      O & O & \cdots & \mathbf{A}_t^{1} & O & O & \cdots & \mathbf{A}_t^{2}\\
      O & O & \cdots & O & O & O & \cdots & O \\
      \vdots & \vdots & \ddots & \vdots & \vdots & \vdots & \ddots & \vdots \\
      \mathbf{B}_1 & \mathbf{B}_2^{1} & \cdots & \mathbf{B}_t^{1} & O & O & \cdots & O
      \end{array}\right).
      \end{equation}
      Since elementary row operations don't affect linear relations among columns, therefore for $\mathbf{H}_2''$, (\ref{eq003}) turns into
      \begin{equation}\label{eq004}
      \sum_{i=1}^{t}\sum_{j\in E_i^{1}}\lambda'_{i,j}\mathbf{h}_{i,j}''+\sum_{i=1}^{t}\sum_{j\in E_i\setminus E_{i}^{1}}\lambda_{i,j}\mathbf{h}_{i,j}''=\mathbf{0},
      \end{equation}
      where $\mathbf{h}_{i,j}''$s are the new columns in $\mathbf{H}_2''$: for $(i,j)\in [t]\times E_{i}^{1}$,
      \begin{equation*}
      \mathbf{h}_{i,j}''=\begin{cases}
      \mathbf{h}_{i,j},~\text{if}~i=1;\\
      \mathbf{h}_{i,j}+(1,g_{i,j},\ldots,g_{i,j}^{\delta-2},0,\ldots,0)^{T},~\text{otherwise};
      \end{cases}
      \end{equation*}
      and for $(i,j)\in [t]\times E_{i}\setminus E_{i}^{1}$,
      \begin{equation*}
      \mathbf{h}_{i,j}''=\begin{cases}
      \mathbf{h}_{i,j}'+(1,g_{i,j},\ldots,g_{i,j}^{\delta-2},0,\ldots,0)^{T},~\text{if}~g_{i,j}\in F_{1}^{1};\\
      \mathbf{h}_{i,j}',~\text{otherwise}.
      \end{cases}
      \end{equation*}
\end{itemize}

Consider the following submatrix consisting of the first $\delta-1$ rows and the last $d-\delta$ rows of $\mathbf{H}_{L}''$:
\begin{equation*}\label{sub_matrix10}
\mathbf{H}_0=\left(\begin{array}{cccc}
\mathbf{A}_1 & \mathbf{A}_2^{1} & \cdots & \mathbf{A}_t^{1}\\
\mathbf{B}_1 & \mathbf{B}_2^{1} & \cdots & \mathbf{B}_t^{1}
\end{array}\right).
\end{equation*}
Clearly, $\mathbf{H}_{0}$ is a $(d-1)\times |F|$ Vandermonde matrix and the construction of $\mathbf{A}_{i}^{1}$ guarantees that columns in $\mathbf{H}_0$ are pairwise distinct. Since $|F|\leq d-1$, it follows that columns in $\mathbf{H}_0$ are linearly independent. On the other hand, according to the structure of $\mathbf{h}_{i,j}''$ for $(i,j)\in [t]\times E_{i}^{1}$, (\ref{eq004}) indicates that
\begin{equation*}
\sum_{i=1}^{t}\sum_{j\in E_i^{1}}\lambda'_{i,j}(1,g_{i,j},\ldots,g_{i,j}^{\delta-2},g_{i,j}^{\delta-1},\ldots,g_{i,j}^{d-2})^{T}=\mathbf{0}.
\end{equation*}
Therefore, we have $\lambda'_{i,j}=0$ for every $(i,j)\in [t]\times E_{i}^{1}$. This leads to
\begin{equation}\label{eq005}
\sum_{i=1}^{t}\sum_{j\in E_i\setminus E_{i}^{1}}\lambda_{i,j}\mathbf{h}_{i,j}''=\mathbf{0},
\end{equation}
where $\lambda_{i,j}$s are the original non-zero coefficients in (\ref{eq002}).

Now, in the following context, based on (\ref{eq005}), we shall derive a contradiction by estimating $\sum_{i=2}^{t}|F_i\setminus F_{i}^{1}|$.

For $2\leq i\leq t$ with $E_i\neq E_{i}^{1}$ and $1\leq l\leq |E_i\setminus E_{i}^{1}|$, let $\mathbf{h}_{i,j_l}''$ be the $l_{th}$ column in
\begin{equation*}
(O~(-\mathbf{A}_{i,2}^{2})^{T}~\cdots~(-\mathbf{A}_{i,i-1}^{2})^{T}~(-\mathbf{A}_{i}^{2})^{T}~O~\cdots~O)^{T},
\end{equation*}
i.e., the $(i-1)_{th}$ block of $\mathbf{H}_R''$. For simplicity of presentation, we rewrite (\ref{eq005}) in the following form:
\begin{equation}\label{eq006}
\left(\begin{array}{cccc}
\mathbf{A}_2^{2} & -\mathbf{A}_{3,2}^{2} & \cdots & -\mathbf{A}_{t,2}^{2}\\
O & \mathbf{A}_3^{2} & \cdots & -\mathbf{A}_{t,3}^{2}\\
\vdots & \vdots & \ddots & \vdots\\
O & O & \cdots & \mathbf{A}_t^{2}
\end{array}\right)\cdot
\left(\begin{array}{c}
\mathbf{v}_{2}^{T}\\
\mathbf{v}_{3}^{T} \\
\vdots\\
\mathbf{v}_{t}^{T}
\end{array}\right)
=\mathbf{0},
\end{equation}
where $\mathbf{v}_i=(\mu_{i,1},\mu_{i,2},\ldots,\mu_{i,|E_i\setminus E_{i}^1|})$ and $\mu_{i,l}=\lambda_{i,j_l}\neq 0$. Given $2\leq i'\leq t$, for $1\leq i\leq i'-1$, define
\begin{equation*}
\mathbf{v}_{i',i}(l)=\begin{cases}
\mathbf{v}_{i'}(l),~\text{if the $l_{th}$ column in $\mathbf{A}_{i',i}^{2}$ is a non-zero vector};\\
0, ~\text{otherwise}.
\end{cases}
\end{equation*}
Since $\sum_{i=1}^{i'-1}\mathbf{A}_{i',i}^{2}=\mathbf{A}_{i'}^{2}$, thus we also have $\sum_{i=1}^{i'-1}\mathbf{v}_{i',i}=\mathbf{v}_{i'}$.
With the help of this observation, (\ref{eq006}) is actually the following system of equations:
\begin{equation}\label{eq007}
\mathbf{A}_i^{2}\cdot \mathbf{v}_{i}^{T}-\sum_{i'>i}\mathbf{A}_{i',i}^{2}\cdot \mathbf{v}_{i',i}^{T}=\mathbf{0},~2\leq i\leq t.
\end{equation}
According to the constructions of $\mathbf{A}_{i}^{2}$ and $\mathbf{A}_{i',i}^{2}$, columns in $\mathbf{A}_{i}^{2}$ are distinct from those columns in $\mathbf{A}_{i',i}^{2}$ for all $i'>i$. Note that $\mu_{i,j}\neq {0}$ for every $2\leq i\leq t$ and $1\leq j\leq |E_i\setminus E_{i}^1|$. Therefore, despite the fact that there might be identical columns in different $\mathbf{A}_{i',i}^{2}$s, (\ref{eq007}) and the property of Vandermonde matrix force that
\begin{equation*}\label{eq008}
\omega(\mathbf{v}_{i})+\sum_{i<i'\leq t}\omega(\mathbf{v}_{i',i})\geq \delta,
\end{equation*}
for every $2\leq i\leq t$. Therefore, we further have
\begin{equation}\label{eq009}
\sum_{i=2}^{t}(\omega(\mathbf{v}_{i})+\sum_{i<i'\leq t}\omega(\mathbf{v}_{i',i}))\geq (t-1)\delta.
\end{equation}
Denote $\mathbf{v}=(\mathbf{v}_{2},\mathbf{v}_{3},\ldots,\mathbf{v}_{t})$. Note that $\sum_{i=2}^{t}|F_i\setminus F_{i}^{1}|=\omega(\mathbf{v})$ and the LHS of (\ref{eq009}) is actually $2\omega(\mathbf{v})-\sum_{i=2}^{t}\omega(\mathbf{v}_{i,1})$, thus we have
\begin{equation*}\label{eq010}
\sum_{i=2}^{t}|F_i\setminus F_{i}^{1}|\geq\frac{\delta}{2}(t-1)+\frac{\sum_{i=2}^{t}\omega(\mathbf{v}_{i,1})}{2}.
\end{equation*}

On the other hand, for each $g_{i,j}\in F$, let $c(g_{i,j})=|\{i'\in[t]: g_{i,j}\in F_{i'}\setminus F_{i'}^{1}\}|$. Through a simple double counting argument, we have
\begin{equation*}
\sum_{g_{i,j}\in F}c(g_{i,j})=\sum_{i=2}^{t}|F_i\setminus F_{i}^{1}|.
\end{equation*}
\begin{itemize}
  \item When $\sum_{i=2}^{t}\omega(\mathbf{v}_{i,1})=0$, we have $\omega(\mathbf{v}_{i,1})=0$ for each $2\leq i\leq t$. This indicates that $(F_{i}\setminus F_{i}^{1})\cap F_{1}=\emptyset$ for each $2\leq i\leq t$, which further leads to $F_1\cap \bigcup_{i=2}^{t}F_i=\emptyset$. Since $F_i\subseteq G_i$ for each $i\in [t]$, thus, we have
      \begin{equation*}
      |\bigcup_{i=2}^{t}G_i|\leq \sum_{i=2}^{t}|G_i|-\sum_{g_{i,j}\in \bigcup_{i=2}^{t}F_i}c(g_{i,j})=(r+\delta-1)(t-1)-\sum_{i=2}^{t}|F_i\setminus F_{i}^{1}|.
      \end{equation*}
      Combined with $|\bigcup_{i=2}^{t}G_l|\geq (r+\frac{\delta}{2}-1)(t-1)+\frac{\delta}{2}$, this leads to $\sum_{i=2}^{t}|F_i\setminus F_{i}^{1}|\leq \frac{\delta}{2}(t-2)$, a contradiction.
  \item When $\sum_{i=2}^{t}\omega(\mathbf{v}_{i,1})>0$, consider $\bigcup_{i\in [t]}G_i$, we have
      \begin{equation*}
      |\bigcup_{i\in [t]}G_i|\leq \sum_{i\in [t]}|G_i|-\sum_{g_{i,j}\in F}c(g_{i,j})=(r+\delta-1)t-\sum_{i=2}^{t}|F_i\setminus F_{i}^{1}|.
      \end{equation*}
      Combined with $|\bigcup_{i\in [t]}G_l|\geq (r+\frac{\delta}{2}-1)t+\frac{\delta}{2}$, this leads to $\sum_{i=2}^{t}|F_i\setminus F_{i}^{1}|\leq \frac{\delta}{2}(t-1)$, a contradiction.
\end{itemize}

Therefore, any $d-1$ columns are linearly independent. This completes the proof of Theorem \ref{main0}.
\end{proof}

\begin{rem}\label{rem1}
\begin{itemize}
  \item[(i)] Theorem \ref{main0} can be viewed as a generalization of Theorem 3.1 in \cite{XY19}, when we take $\delta=2$, the sufficient part of Theorem 3.1 in \cite{XY19} follows from Theorem \ref{main0}.
  \item[(ii)] In \cite{Zhang20}, the author proved a similar result under the condition that
  \begin{equation}\label{cond_zhang}
  |\bigcup_{i\in S}G_i|\geq (r+\delta-2)|S|+1,
  \end{equation}
  for any $S\subseteq[m]$ with $2\leq |S|\leq\lfloor\frac{d-1}{\delta}\rfloor$. Compared to this condition, (\ref{main_cond}) is more relaxed and weakens the restriction of intersections among different repair groups.
\end{itemize}
\end{rem}

\section{Constructions of $(r,\delta)$-LRCs with information locality}\label{sec4}

In this section, we consider linear codes with information $(r,\delta)$-locality. Similarly, through parity-check matrix approach, we provide a general construction for optimal $[n,k,d; (r,\delta)_i]$-LRCs. 

\subsection{Construction B}\label{ssec41}

Let $1\leq v\leq r$, $R=r+\delta-1$ and $n=(l+1)R+h+v-r$ with $h\geq 0$. Let $G_{l+2}=\{g_{l+2,1},g_{l+2,2},\ldots,g_{l+2,h}\}$ be an $h$-subset of $\mathbb{F}_q$, $G_i=\{g_{i,1},g_{i,2},\ldots,g_{i,R}\}$ for $1\leq i\leq l$ and $G_{l+1}=\{g_{l+1,1},g_{l+1,2},\ldots,g_{l+1,v+\delta-1}\}$ be other $l+1$ subsets of $\mathbb{F}_q\setminus G_{l+2}$. Define $f(x)=\prod_{i=1}^{h}(x-g_{l+2,i})$ and consider the following $((l+1)(\delta-1)+h)\times n$ matrix:
\begin{equation}\label{2parity_check_matrix}
\mathbf{H}=\left(\begin{array}{ccccc}
\mathbf{U}_1 &O &\cdots &O & O\\
O &\mathbf{U}_2 &\cdots &O & O\\
\vdots &\vdots &\ddots &\vdots & \vdots\\
O &O &\cdots &\mathbf{U}_{l+1} & O\\
\mathbf{V}_1 &\mathbf{V}_2 &\cdots &\mathbf{V}_{l+1} & \mathbf{V}_{l+2}
\end{array}\right),
\end{equation}
where for $1\leq i\leq l+1$,
\begin{equation*}\label{2pcm1}
\mathbf{U}_i= \left(\begin{array}{cccc}
f(g_{i,1}) &f(g_{i,2}) &\cdots &f(g_{i,|G_{i}|})\\
g_{i,1}f(g_{i,1}) &g_{i,2}f(g_{i,2}) &\cdots &g_{i,|G_{i}|}f(g_{i,|G_{i}|})\\
\vdots &\vdots &\ddots &\vdots\\
g^{\delta-2}_{i,1}f(g_{i,1}) &g^{\delta-2}_{i,2}f(g_{i,2}) &\cdots &g^{\delta-2}_{i,|G_{i}|}f(g_{i,|G_{i}|})
\end{array} \right),~
\mathbf{V}_i= \left(\begin{array}{cccc}
1 &1 &\cdots &1\\
g_{i,1} &g_{i,2} &\cdots &g_{i,|G_{i}|}\\
\vdots &\vdots &\ddots &\vdots\\
g^{h-1}_{i,1} &g^{h-1}_{i,2} &\cdots &g^{h-1}_{i,|G_{i}|}
\end{array} \right)
\end{equation*}
and
\begin{equation*}\label{2pcm2}
\mathbf{V}_{l+2}= \left(\begin{array}{cccc}
1 &1 &\cdots &1\\
g_{l+2,1} &g_{l+2,2} &\cdots &g_{l+2,h}\\
\vdots &\vdots &\ddots &\vdots\\
g^{h-1}_{l+2,1} &g^{h-1}_{l+2,2} &\cdots &g^{h-1}_{l+2,h}
\end{array} \right).
\end{equation*}
Let $\mathcal{C}$ be the $[n,k]$ linear code with parity-check matrix $\mathbf{H}$. Note that for each $1\leq i\leq l$, $\mathbf{U}_i$ is a $(\delta-1)\times R$ matrix with rank $\delta-1$ and $U_{l+1}$ is a $(\delta-1)\times(v+\delta-1)$ matrix with rank $\delta-1$. Therefore, $U_i$s can be viewed as parity-check matrices of generalized Reed-Solomon codes which guarantee that for $1\leq i\leq n-h$, each code symbol $c_i$ has $(r,\delta)$-locality. On the other hand, $\mathcal{C}$ has dimension
\begin{equation*}
k(\mathcal{C})\geq n-(\delta-1)(l+1)-h=lr+v.
\end{equation*}
Since $1\leq v\leq r$, we have $\lceil\frac{k(\mathcal{C})}{r}\rceil\geq l+1$. Thus, $\mathcal{C}$ has minimum distance
\begin{equation}\label{mindis}
d(\mathcal{C})\leq n-k(\mathcal{C})+\delta-\lceil\frac{k(\mathcal{C})}{r}\rceil(\delta-1)\leq h+\delta.
\end{equation}
Therefore, in order to obtain an optimal LRC with $(r,\delta)_i$-locality from the above construction, it suffices to show that the minimum distance of $\mathcal{C}$ equals to $h+\delta$. The same as Section III, our following aim is to find $l+2$ subsets $G_1,G_2,\cdots,G_{l+2}$ in $\mathbb{F}_q$ such that any $h+\delta-1$ columns from the matrix $\mathbf{H}$ are linearly independent.

\subsection{Optimal LRCs with $(r,\delta)_i$-locality from Construction B}\label{ssec42}

In this subsection, sufficient conditions on generating sets $G_1, G_2, \cdots, G_{l+2}$ in Construction B are discussed to guarantee the optimality of the minimum distance. Actually, as we shall see later, $\mathcal{C}$ can recover more than $h+\delta-1$ erasures under proper restrictions on $G_i$s.

For convenience, we use the evaluation points (instead of the indices of code symbols) to denote erasure patterns. Denote $\mathcal{E}=\{E_1,\ldots,E_{l+2}\}$ as an erasure pattern, where $E_i\subseteq G_i$ corresponding to the set of erasure points in $G_i$, $1\leq i\leq l+2$.

\begin{thm}\label{main1}
Let $\mathcal{C}$ be the linear code with parity-check matrix $\mathbf{H}$ from construction B. Let $\mathcal{E}=\{E_1,\ldots,E_{l+2}\}$ be an erasure pattern with $E_i\subseteq G_i$ for $1\leq i\leq l+2$. Denote $S=\{i\in [l+1]: |E_i|\geq \delta\}$. If the erasure pattern $\mathcal{E}$ satisfies
\begin{equation}\label{ineq_erasurepattern}
|\bigcup_{i\in S}E_i|+|E_{l+2}|\leq h+\delta-1
\end{equation}
and
\begin{equation}\label{ineq_evaluationpoint}
|\bigcup_{i\in S}G_i|\geq
\begin{cases}
(r+\frac{\delta}{2}-1)|S|+\frac{\delta}{2},~\text{if } l+1\notin S;\\
(r+\frac{\delta}{2}-1)|S|+\frac{\delta}{2}+v-r,~\text{otherwise},\\
\end{cases}
\end{equation}
then the erasure pattern $\mathcal{E}$ can be recovered.
\end{thm}

\begin{proof}
Note that for any $\mathbf{c}=(c_1,\ldots,c_n)\in \mathcal{C}$ and each $1\leq i\leq n-h$, code symbol $c_i$ in $\mathcal{C}$ has $(r,\delta)$-locality. Therefore, $\mathcal{C}$ is capable of recovering all the erasures in $E_i\in \mathcal{E}$ with $|E_i|\leq \delta-1$ for every $1\leq i\leq l+1$ independently. Thus, we only need to consider erasures from $E_i\in\mathcal{E}$ with $|E_{i}|\geq \delta$ and the erasures from $E_{l+2}$. Let $s=|S|$, w.l.o.g., we can assume that $S=[s]$.

Take $\mathcal{E}'=\{E_1,\ldots,E_s,E_{l+2}\}$ and define $\mathbf{H}|_{\mathcal{E}'}$ as
\begin{equation*}\label{2pcm3}
\mathbf{H}|_{\mathcal{E}'}=\left(\begin{array}{ccccc}
\mathbf{U}_1|_{E_1} &O &\cdots &O & O\\
O &\mathbf{U}_2|_{E_2} &\cdots &O & O\\
\vdots &\vdots &\ddots &\vdots & \vdots\\
O &O &\cdots &\mathbf{U}_{s}|_{E_{s}} & O\\
\vdots &\vdots &\vdots &\vdots & \vdots\\
\mathbf{V}_1|_{E_1} &\mathbf{V}_2|_{E_2} &\cdots &\mathbf{V}_{s}|_{E_{s}} & \mathbf{V}_{l+2}|_{E_{l+2}}
\end{array}\right),
\end{equation*}
for simplicity of presentation, here, $U_i|_{E_i}$ ($V_{i}|_{E_i}$) denotes the restriction of $U_i$ ($V_i$) to the set of columns generated by elements in $E_i$. Note that an erasure pattern $\mathcal{E}$ can be recovered by the code $\mathcal{C}$ with parity-check matrix $\mathbf{H}$ if and only if the restriction of $\mathbf{H}$ to $\mathcal{E}$ has full column rank. Therefore, we only need to show that $\mathbf{H}|_{\mathcal{E}'}$ has full column rank.

Assume not, i.e., there exists a non-zero vector $\mathbf{v}=(\mathbf{v}_1,\mathbf{v}_2,\ldots,\mathbf{v}_{s},\mathbf{v}_{l+2})$ with $\mathbf{v}_{i}\in \mathbb{F}_q^{|E_i|}$ such that
\begin{equation}\label{eq021}
\mathbf{H}|_{\mathcal{E}'}\cdot
\left(\begin{array}{c}
\mathbf{v}_1^{T}\\
\vdots\\
\mathbf{v}_{s}^{T}\\
\mathbf{v}_{l+2}^{T}
\end{array}\right)
=\mathbf{0}.
\end{equation}
Write $\mathbf{H}|_{E_i}=(\mathbf{h}_{i,1}~\mathbf{h}_{i,2}~\cdots~\mathbf{h}_{i,|E_i|})$, $\mathbf{v}_{i}=(\lambda_{i,1},\ldots,\lambda_{i,|E_i|})$ and assume that $\mathbf{h}_{i,j}$ is generated by $a_{i,j}\in E_{i}\subseteq G_{i}$. Denote $E_i'=\{a_{i,j}:a_{i,j}\in E_{i} \text{~and~} \lambda_{i,j}\neq 0\}$, $\mathcal{E}''=\{E_1',\ldots,E_s',E_{l+2}'\}$ and $\mathbf{v}_i'$ as the vector of length $\omega(\mathbf{v}_i)$ by puncturing $\mathbf{v}_i$ on its non-zero coordinates.
Then, (\ref{eq021}) turns into the following form
\begin{equation*}\label{eq022}
\mathbf{H}|_{\mathcal{E}''}\cdot
\left(\begin{array}{c}
(\mathbf{v}_1')^{T}\\
\vdots\\
(\mathbf{v}_{s}')^{T}\\
(\mathbf{v}_{l+2}')^{T}
\end{array}\right)
=\mathbf{0}.
\end{equation*}
This reduces the problem to a sub-erasure pattern $\mathcal{E}''$ of $\mathcal{E}'$. Thus, w.l.o.g., we can assume that $\lambda_{i,j}\neq 0$ for every $i\in [s]\cup\{l+2\}$ and $1\leq j\leq |E_{i}|$.

For $i\in [s]$, let
\begin{equation*}
\mathbf{A}_i=\mathbf{U}_i|_{E_i}=\left(\begin{array}{cccc}
f(a_{i,1}) &f(a_{i,2}) &\cdots &f(a_{i,|E_i|})\\
a_{i,1}f(a_{i,1}) &a_{i,2}f(a_{i,2}) &\cdots &a_{i,|E_i|}f(a_{i,|E_i|})\\
\vdots &\vdots &\ddots &\vdots\\
a^{\delta-2}_{i,1}f(a_{i,1}) &a^{\delta-2}_{i,2}f(a_{i,2}) &\cdots &a^{\delta-2}_{i,|E_i|}f(a_{i,|E_i|})
\end{array}\right)
\end{equation*}
and for $i\in [s]\cup\{l+2\}$, let
\begin{equation*}
\mathbf{B}_i=\mathbf{V}_i|_{E_i}=\left(\begin{array}{cccc}
1 &1 &\cdots &1\\
a_{i,1} &a_{i,2} &\cdots &a_{i,|E_i|}\\
\vdots &\vdots &\ddots &\vdots\\
a^{h-1}_{i,1} &a^{h-1}_{i,2} &\cdots &a^{h-1}_{i,|E_{i}|}
\end{array}\right).
\end{equation*}
Denote $E_1^{1}=E_1$, $E_i^{1}=E_{i}\setminus \bigcup_{j=1}^{i-1}E_i$ for $2\leq i\leq s$ and $E=\sqcup_{i=1}^{s}E_{i}^{1}$. By permutating the columns of $\mathbf{H}|_{\mathcal{E}'}$, we can obtain an equivalent matrix of the following form:
\begin{equation*}\label{2pcm5}
\mathbf{H}_2=(\mathbf{H}_L || \mathbf{H}_R)=\left(\begin{array}{ccccc||cccc}
\mathbf{A}_1 & O & \cdots & O & O & O & O & \cdots & O\\
O & \mathbf{A}_2^{1} & \cdots & O & O & \mathbf{A}_2^{2} & O & \cdots & O\\
O & O & \cdots & O & O & O & \mathbf{A}_3^{2} & \cdots & O\\
\vdots & \vdots & \ddots & \vdots & \vdots & \vdots & \vdots & \ddots & \vdots\\
O & O & \cdots & \mathbf{A}_s^{1} & O & O & O & \cdots & \mathbf{A}_s^{2}\\
O & O & \cdots & O & O & O & O & \cdots & O \\
\vdots & \vdots & \ddots & \vdots & \vdots & \vdots & \vdots & \ddots & \vdots \\
\mathbf{B}_1 & \mathbf{B}_2^{1} & \cdots & \mathbf{B}_s^{1} & \mathbf{B}_{l+2} & \mathbf{B}_2^{2} & \mathbf{B}_3^{2} & \cdots & \mathbf{B}_s^{2}
\end{array}\right),
\end{equation*}
where for $2\leq i\leq s$, $\mathbf{A}_i^{1}=\mathbf{A}_i|_{E_{i}^1}$, $\mathbf{A}_i^{2}=\mathbf{A}_i|_{E_{i}\setminus E_{i}^1}$ and $\mathbf{B}_i^{1}=\mathbf{B}_i|_{E_{i}^1}$, $\mathbf{B}_i^{2}=\mathbf{B}_i|_{E_{i}\setminus E_{i}^1}$. For each $i\in [s]$, denote $I_i^{1}=\{j\in [|E_i|]: a_{i,j}\in E_{i}^{1}\}$ and $I_i^{2}=\{j\in [|E_i|]: a_{i,j}\in E_{i}\setminus E_{i}^{1}\}$. Then,  (\ref{eq021}) can be written as
\begin{equation}\label{eq023}
\sum_{i=1}^{s}\sum_{j\in I_i^{1}}\lambda_{i,j}\mathbf{h}_{i,j}+\sum_{j\in [|E_{l+2}|]}\lambda_{l+2,j}\mathbf{h}_{l+2,j}+\sum_{i=1}^{s}\sum_{j\in I_{i}^{2}}\lambda_{i,j}\mathbf{h}_{i,j}=\mathbf{0}.
\end{equation}

Similar to the proof of Theorem \ref{main0}, we can do the following elementary row and column operations:
\begin{itemize}
  \item First, for each $2\leq i\leq s$ and each column $\mathbf{h}_{i,j}$ in $(O~\cdots~O~(\mathbf{A}_i^{2})^{T}~O~\cdots~(\mathbf{B}_{i}^{2})^{T})^{T}$ of $\mathbf{H}_R$, subtract the column $\mathbf{h}_{i',j'}$ in $(O~\cdots~O~(\mathbf{A}_{i'}^{1})^{T}~O~\cdots~(\mathbf{B}_{i'}^{1})^{T})^{T}$ of $\mathbf{H}_L$ from it, where $(i',j')$ satisfies $i'<i$ and $a_{i,j}=a_{i',j'}\in E_{i'}^{1}$. This leads to a matrix equivalent to $\mathbf{H}_2$:
      \begin{equation*}\label{2pcmH_2'}
      \mathbf{H}_2'=(\mathbf{H}_L || \mathbf{H}_R')=\left(\begin{array}{ccccc|cccc}
      \mathbf{A}_1 & O & \cdots & O & O & -\mathbf{A}_{2,1}^{2} & -\mathbf{A}_{3,1}^{2} & \cdots & -\mathbf{A}_{s,1}^{2}\\
      O & \mathbf{A}_2^{1} & \cdots & O & O & \mathbf{A}_2^{2} & -\mathbf{A}_{3,2}^{2} & \cdots & -\mathbf{A}_{s,2}^{2}\\
      O & O & \cdots & O & O & O & \mathbf{A}_3^{2} & \cdots & -\mathbf{A}_{s,3}^{2}\\
      \vdots & \vdots & \ddots & \vdots & \vdots & \vdots & \vdots & \ddots & \vdots\\
      O & O & \cdots & \mathbf{A}_s^{1} & O & O & O & \cdots & \mathbf{A}_s^{2}\\
      O & O & \cdots & O & O & O & O & \cdots & O \\
      \vdots & \vdots & \ddots & \vdots & \vdots & \vdots & \vdots & \ddots & \vdots \\
      \mathbf{B}_1 & \mathbf{B}_2^{1} & \cdots & \mathbf{B}_s^{1} & \mathbf{B}_{l+2} & O & O & \cdots & O
      \end{array}\right),
      \end{equation*}
      where for $1\leq j<i\leq s$, the $l_{th}$ column of $\mathbf{A}_{i,j}^{2}$ is identical to the $l_{th}$ column of $\mathbf{A}_i^{2}$ if the corresponding generating element appears in $E_{j}^{1}\cap E_{i}$ and is identical to the zero vector, otherwise. Clearly, we have $\mathbf{A}_{2,1}^{2}=\mathbf{A}_2^{2}$ and $\sum_{j=1}^{i-1}\mathbf{A}_{i,j}^{2}=\mathbf{A}_i^{2}$. Moreover, (\ref{eq023}) turns into
      \begin{equation}\label{eq024}
      \sum_{i=1}^{s}\sum_{j\in I_i^{1}}\lambda'_{i,j}\mathbf{h}_{i,j}+\sum_{j\in[|E_{l+2}|]}\lambda_{l+2,j}\mathbf{h}_{l+2,j}+\sum_{i=1}^{s}\sum_{j\in I_i^{2}}\lambda_{i,j}\mathbf{h}'_{i,j}=\mathbf{0},
      \end{equation}
      where for $(i,j)\in [s]\times I_{i}^{1}$,
      \begin{equation*}
      \lambda_{i,j}'=\lambda_{i,j}+\sum_{i< i'\leq s}\sum_{\substack{j'\in I_{i'}^{2}:\\a_{i',j'}=a_{i,j}}}\lambda_{i',j'},
      \end{equation*}
      and for $(i,j)\in [s]\times I_{i}^{2}$,
      \begin{equation*}
      \mathbf{h}_{i,j}'=\mathbf{h}_{i,j}-\mathbf{h}_{i',j'}\in \mathbf{H}_{R}'
      \end{equation*}
      for some $(i',j')$ satisfying $i'<i$ and $j'\in I_{i'}^{1}$ such that $a_{i,j}=a_{i',j'}$.
  \item Second, for each $1\leq i\leq \delta-1$, add the $(i+j(\delta-1))_{th}$ row to the $i_{th}$ row for all $1\leq j\leq s-1$. This leads to
      \begin{equation}\label{2pcmH_2''}
      \mathbf{H}_2''=(\mathbf{H}_L'' || \mathbf{H}_R'')=\left(\begin{array}{ccccc||cccc}
      \mathbf{A}_{1} & \mathbf{A}_2^{1} & \cdots & \mathbf{A}_{s}^{1} & O & O & O & \cdots & O\\
      O & \mathbf{A}_2^{1} & \cdots & O & O & \mathbf{A}_2^{2} & -\mathbf{A}_{3,2}^{2} & \cdots & -\mathbf{A}_{s,2}^{2}\\
      O & O & \cdots & O &  O & O & \mathbf{A}_3^{2} & \cdots & -\mathbf{A}_{s,3}^{2}\\
      \vdots & \vdots & \ddots & \vdots & \vdots & \vdots &  \vdots & \ddots & \vdots\\
      O & O & \cdots & \mathbf{A}_s^{1} & O & O & O & \cdots & \mathbf{A}_s^{2}\\
      O & O & \cdots & O & O & O & O & \cdots & O \\
      \vdots & \vdots & \ddots & \vdots & \vdots & \vdots & \vdots & \ddots & \vdots \\
      \mathbf{B}_1 & \mathbf{B}_2^{1} & \cdots & \mathbf{B}_s^{1} & \mathbf{B}_{l+2} & O & O & \cdots & O
      \end{array}\right).
      \end{equation}
      Since elementary row operations don't affect linear relations among columns, therefore for $\mathbf{H}_2''$, (\ref{eq024}) turns into
      \begin{equation}\label{eq025}
      \sum_{i=1}^{s}\sum_{j\in I_i^{1}}\lambda'_{i,j}\mathbf{h}_{i,j}''+\sum_{j\in[|E_{l+2}|]}\lambda_{l+2,j}\mathbf{h}_{l+2,j}+\sum_{i=1}^{s}\sum_{j\in I_i^{2}}\lambda_{i,j}\mathbf{h}_{i,j}''=\mathbf{0},
      \end{equation}
      where for $(i,j)\in [s]\times I_{i}^{1}$,
      \begin{equation*}
      \mathbf{h}_{i,j}''=\begin{cases}
      \mathbf{h}_{i,j},~\text{when}~i=1;\\
      \mathbf{h}_{i,j}+(f(a_{i,j}),a_{i,j}f(a_{i,j}),\ldots,a_{i,j}^{\delta-2}f(a_{i,j}),0,\ldots,0)^{T},~\text{when}~i\geq 2,
      \end{cases}
      \end{equation*}
      and for $(i,j)\in [s]\times I_{i}^{2}$,
      \begin{equation*}
      \mathbf{h}_{i,j}''=\begin{cases}
      \mathbf{h}_{i,j}'+(f(a_{i,j}),a_{i,j}f(a_{i,j}),\ldots,a_{i,j}^{\delta-2}f(a_{i,j}),0,\ldots,0)^{T},~\text{if}~a_{i,j}\in E_{1}^{1};\\
      \mathbf{h}_{i,j}',~\text{otherwise}.
      \end{cases}
      \end{equation*}
\end{itemize}

Now, consider the following submatrix consisting of the first $\delta-1$ rows and the last $h$ rows of $\mathbf{H}_{L}''$:  \begin{equation*}\label{2sub_matrix1}
\mathbf{H}_0=\left(\begin{array}{ccccc}
\mathbf{A}_1 & \mathbf{A}_2^{1} & \cdots & \mathbf{A}_s^{1} & O\\
\mathbf{B}_1 & \mathbf{B}_2^{1} & \cdots & \mathbf{B}_s^{1} & \mathbf{B}_{l+2}
\end{array}\right).
\end{equation*}
Clearly, $\mathbf{H}_0$ is of size $(h+\delta-1)\times (|E|+|E_{l+2}|)$.
\begin{claim}\label{claim1}
$\mathbf{H}_0$ has full column rank.
\end{claim}
\begin{proof}
Since $f(g_{l+2,j})=0$ for every $g_{l+2,j}\in G_{l+2}$, thus, we can treat the zero submatrix in the top-right corner of $\mathbf{H}_{0}$ as:
\begin{equation*}
\mathbf{A}_{l+2}=\left(\begin{array}{cccc}
f(a_{l+2,1}) &f(a_{l+2,2}) &\cdots &f(a_{l+2,|E_{l+2}|})\\
a_{l+2,1}f(a_{l+2,1}) &a_{l+2,2}f(a_{l+2,2}) &\cdots &a_{l+2,|E_{l+2}|}f(a_{l+2,|E_{l+2}|})\\
\vdots &\vdots &\ddots &\vdots\\
a^{\delta-2}_{l+2,1}f(a_{l+2,1}) &a^{\delta-2}_{l+2,2}f(a_{l+2,2}) &\cdots &a^{\delta-2}_{l+2,|E_{l+2}|}f(a_{l+2,|E_{l+2}|})
\end{array}\right),
\end{equation*}
where $\{a_{l+2,1},a_{l+2,2},\ldots,a_{l+2,|E_{l+2}|}\}= E_{l+2}$.

For ease of notations, we let $E_{l+2}^{1}=E_{l+2}$. Note that
\begin{equation*}
|E|+|E_{l+2}|=|\bigcup_{i\in [s]}E_i|+|E_{l+2}|\leq h+\delta-1,
\end{equation*}
when $|E|+|E_{l+2}|\geq h+1$, we can consider the square submatrix of $\mathbf{H}_0$ consisting of the first $d_0-h+1=|E|+|E_{l+2}|-h$ rows and the last $h$ rows:
\begin{equation*}
\mathbf{H}_0'=\left(\begin{array}{ccccc}
\tilde{\mathbf{A}}_1 & \tilde{\mathbf{A}}_2 & \cdots & \tilde{\mathbf{A}}_s & \tilde{\mathbf{A}}_{l+2}\\
\mathbf{B}_1 & \mathbf{B}_2^{1} & \cdots & \mathbf{B}_s^{1} & \mathbf{B}_{l+2}
\end{array}\right),
\end{equation*}
where for $i\in[s]\cup\{l+2\}$,
\begin{equation*}
\tilde{\mathbf{A}}_{i}=\left(\begin{array}{cccc}
f(a_{i,1}) & f(a_{i,2}) &\cdots & f(a_{i,|E_{i}^{1}|})\\
a_{i,1}f(a_{i,1}) & a_{i,2}f(a_{i,2}) &\cdots & a_{i,|E_{i}^{1}|}f(a_{i,|E_{i}^{1}|})\\
\vdots & \vdots & \ddots &\vdots\\
a^{d_0-h}_{i,1}f(a_{i,1}) & a^{d_0-h}_{i,2}f(a_{i,2}) & \cdots & a^{d_0-h}_{i,|E_{i}^{1}|}f(a_{i,|E_{i}^{1}|})
\end{array}\right).
\end{equation*}
For any integer $d\geq 0$, denote $\mathbb{F}_{q}^{\leq d}[x]$ as the linear space of polynomials with degree at most $d$ in $\mathbb{F}_{q}[x]$. Since $\{1,x,\ldots,x^{h-1}\}$ together with $\{f(x),xf(x),\ldots,x^{d_0-h}f(x)\}$ form a basis of $\mathbb{F}_{q}^{\leq d_0}[x]$, therefore, for any non-zero vector $\mathbf{u}\in \mathbb{F}_{q}^{d_0+1}$,
\begin{equation*}
\mathbf{u}\cdot (f(x),xf(x),\ldots,x^{d_0-h}f(x),1,x,\ldots,x^{h-1})^{T}\in \mathbb{F}_{q}^{\leq d_0}[x]
\end{equation*}
has at most $d_0=|E|+|E_{l+2}|-1$ different zeros in $\mathbb{F}_q$. Since $a_{i,j}$s from $\bigsqcup_{i=1}^{s}E_{i}^{1}\sqcup E_{l+2}$ are pairwise distinct, thus, $\mathbf{u}\cdot\mathbf{H}_0'\neq \mathbf{0}$ for any non-zero vector $\mathbf{u}\in \mathbb{F}_{q}^{d_0+1}$. This shows that $rank(\mathbf{H}_0')=|E|+|E_{l+2}|$.

When $|E|+|E_{l+2}|\leq h$, consider the square submatrix $\mathbf{H}_0'$ consisting of the last $|E|+|E_{l+2}|$ rows of $\mathbf{H}_0$. Similarly, by the property of Vandermonde-type matrices, we can also obtain $rank(\mathbf{H}_0')=|E|+|E_{l+2}|$.

To sum up, for both cases, $\mathbf{H}_0$ contains a square submatrix of rank $|E|+|E_{l+2}|$, therefore, $\mathbf{H}_0$ has full column rank.
\end{proof}

According to the structure of $\mathbf{h}_{i,j}''$, (\ref{eq025}) together with Claim \ref{claim1} actually indicates that
\begin{equation*}
\begin{cases}
\lambda'_{i,j}=0,~\text{for}~(i,j)\in [s]\times I_{i}^{1};\\
E_{l+2}=\emptyset.
\end{cases}
\end{equation*}
Therefore, we have
\begin{equation}\label{eq026}
\sum_{i=1}^{s}\sum_{j\in I_i^{2}}\lambda_{i,j}\mathbf{h}_{i,j}''=\mathbf{0}.
\end{equation}
Note that the structures of $\mathbf{H}_{R}''$s in (\ref{H_2''0}) and (\ref{2pcmH_2''}) are the same and $\mathbf{A}_i$s here are also Vandermonde-like matrices. Therefore, through an analogous argument to the latter part of the proof of Theorem \ref{main0}, we can also derive a contradiction by estimating $\sum_{i=2}^{t}|E_i\setminus E_{i}^{1}|$, which shows that $\mathbf{H}|_{\mathcal{E}'}$ has full column rank.

This completes the proof of Theorem \ref{main1}.
\end{proof}

In the same vine, for small $h$, we can obtain optimal $(r,\delta)_i$-LRCs with arbitrarily long length by Theorem \ref{main1}.
\begin{thm}\label{thm41}
Let $1\leq h\leq \delta$, set $G_1=G_2=\cdots=G_{l}$ and $G_{l+1}$ as any $v+\delta-1$-subset of $G_1$ in Construction B, then the code $\mathcal{C}$ generated by Construction B can correct any $h+\delta-1$ erasures.
\end{thm}
\begin{proof}
Given any erasure pattern $\mathcal{E}=\{E_1,\ldots,E_{l+2}\}$ satisfying $\sum_{i=1}^{l+2}|E_i|=h+\delta-1$. Since $1\leq h\leq \delta$, thus, there is only one block which contains at least $\delta$ columns. This indicates that $|S|=1$ and thus, (\ref{ineq_evaluationpoint}) holds naturally. Therefore, $\mathcal{E}$ can be recovered by $\mathcal{C}$.
\end{proof}
\begin{cor}\label{cor40}
Let $q\geq r+\delta-1$. Assume that $1\leq h\leq \delta$, then there exists an optimal $[n,k,h+\delta;(r,\delta)_i]$-LRC with length $n=(l+1)(r+\delta-1)+h+v-r$ for any positive integer $l$.
\end{cor}

As another consequence of Theorem \ref{main1}, for general $h$, we have the following corollary.

\begin{cor}\label{cor41}
If the system $\mathcal{G}=\{G_1,\ldots,G_{l+1}\}$ from Construction B satisfies
\begin{equation}\label{ineq_evaluationpoint2}
|\bigcup_{i\in S}G_i|\geq
\begin{cases}
(r+\frac{\delta}{2}-1)|S|+\frac{\delta}{2},~\text{if } l+1\notin S;\\
(r+\frac{\delta}{2}-1)|S|+\frac{\delta}{2}+v-r,~\text{otherwise},
\end{cases}
\end{equation}
for every subset $S\subseteq [l+1]$ of size at most $\lfloor\frac{h+\delta-1}{\delta}\rfloor$, then the code $\mathcal{C}$ generated by Construction B is an optimal $[n,k,h+\delta;(r,\delta)_i]$-LRC.
\end{cor}

\begin{proof}
According to (\ref{mindis}), we only need to show that the code $\mathcal{C}$ can recover any erasure pattern $\mathcal{E}=\{E_i:1\leq i\leq l+2\}$ with $\sum_{i\in S}|E_i|+|E_{l+2}|\leq h+\delta-1$, where $E_i\subseteq G_i$. For any $\mathbf{c}=(c_1,\ldots,c_n)\in \mathcal{C}$ and $1\leq j\leq n-h$, the structure of $\mathbf{H}$ in Construction B ensures that the code symbol $c_j$ has $(r,\delta)$-locality. Therefore, for $i\in S$ with $|E_{i}|<\delta$, $E_i$ can be recovered. Denote $S'=\{i'\in S: |E_{i'}|\geq \delta\}$. Then, we have $|S'|\leq \lfloor\frac{h+\delta-1}{\delta}\rfloor$ and
\begin{equation*}
\sum_{i'\in S'}|E_{i'}|+|E_{l+2}|\leq h+\delta-1.
\end{equation*}
This leads to $|\bigcup_{i'\in S'}E_{i'}|+|E_{l+2}|\leq h+\delta-1$ and the result follows from Theorem \ref{main1}.
\end{proof}

\begin{rem}\label{rem2}
In \cite{CS20}, based on ideas of polynomial interpolation, Cai and Schwartz provide a construction of LRCs with $(r,\delta)_i$-locality with the same recovering capability (see Theorem 1 in \cite{CS20}). From the perspective of parity-check matrix, Theorem 1 in \cite{CS20} requires the generating sets $G_i$s to satisfy
\begin{equation*}
|G_{i}\cap(\bigcup_{i\neq j\in S}G_{j})|\leq \delta-1,
\end{equation*}
which is a local condition for $G_{i}$s. However, the minimal distance is a global parameter of the code. Therefore, due to the advantage of the intrinsic combinatorial property of $G_i$s satisfying (\ref{ineq_evaluationpoint}), Theorem \ref{main1} and Corollary \ref{cor41} can provide longer codes.
\end{rem}

\section{Optimal LRCs based on sparse hypergraphs}\label{sec5}

\subsection{Tu\'{r}an-type problems for sparse hypergraphs}\label{ssec51}
Throughout this section, we will use some standard notations of sparse hypergraph from \cite{ST20}. An $R$-uniform hypergraph ($R$-graph for short) on $n$ vertices $\mathcal{H}:=(V(\mathcal{H}),E(\mathcal{H}))$ is a pair of vertices and edges, where the vertex set $V(\mathcal{H})$ is a finite set (denoted as $[n]$) and edge set $E(\mathcal{H})$ is a collection of $R$-subsets of $V(\mathcal{H})$. For convenience, we often use $\mathcal{H}$ to denote its edge set $E(\mathcal{H})$ if there is no confusion.

For positive integers $v$ and $e$, let $\mathcal{G}_R(v,e)$ be the family of all $R$-graphs consisting of $e$ edges and at most $v$ vertices, i.e.,
$$\mathcal{G}_R(v,e)=\{\mathcal{H}\subseteq\dbinom{[n]}{R}:|E(\mathcal{H})|=e,|V(\mathcal{H})|\leq v\}.$$
Then, an $R$-graph $\mathcal{H}$ is said to be $\mathcal{G}_R(v,e)$-free if it does not contain a copy of any member in $\mathcal{G}_R(v,e)$. In relevant literatures, such $R$-graphs are called sparse hypergraphs. Usually, we denote $f_R(n,v,e)$
as the maximum number of edges in a $\mathcal{G}_R(v,e)$-free $R$-graph on $n$ vertices.

In \cite{brown1973some}, Brown, Erd\H{o}s and S\'{o}s first made the following estimation about the value of $f_R(n,v,e)$.
\begin{lem}\cite{brown1973some}\label{lem41}
For $R\geq2, e\geq2, v\geq{R+1}$, there exist constants $c_1$, $c_2$ depending only on $R$, $e$, $v$ such that $$c_1n^{\frac{eR-v}{e-1}}\leq f_R(n,v,e)\leq c_2n^{\lceil\frac{eR-v}{e-1}\rceil}.$$
\end{lem}
When $e-1|eR-v$, this already determined the order of $f_R(n,v,e)$ up to a constant factor. However, for $e-1\nmid eR-v$, it turns out to be extremely difficult to determine the correct exponent. With additional condition that $v=3(R-l)+l+1$ and $e=3$, Alon and Shapira \cite{alon2006extremal} proved the next result.
\begin{lem}\cite{alon2006extremal}\label{lem42}
For $2\leq l<R$, we have  $$n^{l-o(1)}<f_R(n,3(R-l)+l+1,3)=o(n^l).$$
Furthermore, there exists an explicit construction of $R$-graph which is both $\mathcal{G}_R(3(R-l)+l+1,3)$-free and $\mathcal{G}_R(2(R-l)+l,2)$-free with $n^{l-o(1)}$ edges.
\end{lem}
Later in 2017, Ge and Shangguan \cite{ge2017sparse} provided a construction for hypergraphs forbidding small rainbow cycles with order-optimal edges w.r.t. Lemma \ref{lem41} (see Theorem 1.6 in \cite{ge2017sparse}). For general lower bound on $f_R(n,v,e)$, very recently, Shangguan and Tamo \cite{ST20} proved the following result.
\begin{thm}\cite{ST20}\label{lem44}
For $R\geq2, e\geq3, v\geq{R+1}$ satisfying $\emph{gcd}(e-1,eR-v)=1$ and sufficiently large $n$, there exists an $R$-graph with $$\Omega(n^{\frac{eR-v}{e-1}}(\log n)^{\frac{1}{e-1}})$$edges,
which is also $\mathcal{G}_R(iR-\lceil\frac{(i-1)(eR-v)}{e-1}\rceil,i)$-free for every $2\leq i\leq e$, and in particular,
$$f_R(n,v,e)=\Omega(n^{\frac{eR-v}{e-1}}(\log n)^{\frac{1}{e-1}})$$
as $n\rightarrow \infty$. Here the constants in $\Omega(\cdot)$ are independent of $n$.
\end{thm}
In the same paper, Shangguan and Tamo also considered this type of problems for hypergraphs that are simultaneously $\mathcal{G}_{R}(v_i,e_i)$-free for a series of $\{(v_i,e_i)\}_{i=1}^{s}$.
\begin{lem}\cite{ST20}\label{lem45}
Let $s\geq 1$, $R\geq 3$ and $(v_i,e_i)$, $1\leq i\leq s$ be fixed integers satisfying $v_i\geq R+1$, $e_i\geq 2$. Suppose further that $e_1\geq 3$, $\gcd(e_1-1,e_1R-v_1)=1$ and $\frac{e_1R-v_1}{e_1-1}<\frac{e_iR-v_i}{e_i-1}$ for $2\leq i\leq s$. Then there exists an $R$-graph with $\Omega(n^{\frac{e_1R-v_1}{e_1-1}}(\log n)^{\frac{1}{e_1-1}})$ edges which is $\mathcal{G}_{R}(v_i,e_i)$-free for each $1\leq i\leq s$.
\end{lem}

According to Theorem \ref{main0} and Theorem \ref{main1}, constructions of both optimal $(r,\delta)_a$-LRCs and optimal $(r,\delta)_i$-LRCs require the generating sets $G_i$s to form a special kind of sparse hypergraph which is simultaneously $\mathcal{G}_{R}(iR-\lfloor(i-1)\frac{\delta}{2}\rfloor-1,i)$-free for $2\leq i\leq\mu$ (for some given integer $\mu\geq 3$). Armed with the above results, we have the following existence theorem for such hypergraphs with $|E(\mathcal{H})|$ growing super-linearly in $n$.

\begin{thm}\label{main51}
Let $\delta \geq 2$, $\mu\geq 3$ and $R\geq \min\{\delta,3\}$ be fixed integers. Then, for $n$ sufficiently large, there exists an $R$-uniform hypergraph $\mathcal{H}(V, E)$ that is simultaneously $\mathcal{G}_{R}(iR-\lfloor(i-1)\frac{\delta}{2}\rfloor-1,i)$-free for every $2\leq i\leq\mu$ with $|V|=n$ and
\begin{equation}\label{shg01}
|E|=\begin{cases}
\Omega(n^{\frac{\delta}{2}+\frac{1}{\mu-1}}{(\log{n})}^{\frac{1}{\mu-1}}), ~\text{when}~\delta~\text{is even};\\
\Omega(n^{\frac{\delta}{2}+\frac{1}{2(\mu-1)}}{(\log{n})}^{\frac{1}{2(\mu-1)}}), ~\text{when}~\delta~\text{is odd}~{and}~\mu~\text{is even};\\
\Omega(n^{\frac{\delta}{2}+\frac{1}{2(\mu-2)}}{(\log{n})}^{\frac{1}{2(\mu-2)}}), ~\text{when both}~\delta~\text{and}~\mu>3~\text{are odd};\\
\Omega(n^{\frac{\delta+1}{2}}), ~\text{when}~\delta~\text{is odd}~{and}~\mu=3.
\end{cases}
\end{equation}
\end{thm}

\begin{proof}
For each $2\leq i\leq \mu$, let $v_i'=iR-\lfloor(i-1)\frac{\delta}{2}\rfloor-1$. Consider the sequence $\{\frac{iR-v'_i}{i-1}\}_{i=2}^{\mu}=\{\frac{\lfloor(i-1)\frac{\delta}{2}\rfloor+1}{i-1}\}_{i=2}^{\mu}$: For each $2\leq i\leq \mu$,
\begin{equation*}
\frac{\lfloor(i-1)\frac{\delta}{2}\rfloor+1}{i-1}=\begin{cases}
\frac{\delta}{2}+\frac{1}{i-1}, ~\text{if}~(i-1)\delta~\text{is even};\\
\frac{\delta}{2}+\frac{1}{2(i-1)}, ~\text{if}~(i-1)\delta~\text{is odd}.
\end{cases}
\end{equation*}
Therefore, when $\delta$ is even, $\{\frac{\lfloor(i-1)\frac{\delta}{2}\rfloor+1}{i-1}\}_{i=2}^{\mu}=\{\frac{\delta}{2}+\frac{1}{i-1}\}_{i=2}^{\mu}$ is a strictly decreasing sequence and $\frac{\delta}{2}+\frac{1}{\mu-1}<\frac{\delta}{2}+\frac{1}{i-1}$ for all $2\leq i\leq \mu-1$. When $\delta$ is odd, we have
\begin{equation*}
\frac{\lfloor(i-1)\frac{\delta}{2}\rfloor+1}{i-1}=\begin{cases}
\frac{\delta}{2}+\frac{1}{i-1}, ~\text{if}~i~\text{is odd};\\
\frac{\delta}{2}+\frac{1}{2(i-1)}, ~\text{if}~i~\text{is even}.
\end{cases}
\end{equation*}
Therefore, based on the monotone decreasing property of both $\frac{\delta}{2}+\frac{1}{i-1}$ (for odd $i$) and $\frac{\delta}{2}+\frac{1}{2(i-1)}$ (for even $i$), we have
\begin{equation*}
\begin{cases}
\frac{\delta}{2}+\frac{1}{2(\mu-1)}<\frac{\lfloor(i-1)\frac{\delta}{2}\rfloor+1}{i-1} \text{ for } 2\leq i\leq \mu-1 , \text{ when } \mu \text{ is even};\\
\frac{\delta}{2}+\frac{1}{2(\mu-2)}<\frac{\lfloor(i-1)\frac{\delta}{2}\rfloor+1}{i-1} \text{ for } 2\leq i\leq \mu-2 \text{ and } i=\mu,  \text{ when } \mu>3 \text{ and } \mu \text{ is odd};\\
\frac{\delta}{2}+\frac{1}{2(\mu-2)}=\frac{\delta}{2}+\frac{1}{\mu-1}, \text{ when } \mu=3.
\end{cases}
\end{equation*}

When $\delta$ is even, clearly, we have $\gcd(\mu-1,(\mu-1)\frac{\delta}{2}+1)=1$. By applying Lemma \ref{lem45} with $s=\mu-1$, $(v_1,e_1)=(\mu R-(\mu-1)\frac{\delta}{2}-1,\mu)$ and $\{(v_i,e_i)\}_{i=2}^{\mu-1}=\{(v_j',j)\}_{j=2}^{\mu-1}$, there exists an $R$-graph with $\Omega(n^{\frac{\delta}{2}+\frac{1}{\mu-1}}{(\log{n})}^{\frac{1}{\mu-1}})$ edges. This proves the first part of (\ref{shg01}).

When $\delta$ is odd, $\mu$ is even. Assume that $\mu=2u$ for some $u\geq 2$. Then, we have
\begin{equation*}
(\mu-1)\frac{\delta}{2}+\frac{1}{2}=(\mu-1)\frac{\delta-1}{2}+u.
\end{equation*}
Since $\delta$ is odd, thus $(\mu-1)|(\mu-1)\frac{\delta-1}{2}$. Therefore, we have $\gcd(\mu-1,(\mu-1)\frac{\delta}{2}+\frac{1}{2})=1$. By applying Lemma \ref{lem45} with $s=\mu-1$, $(v_1,e_1)=(\mu R-((\mu-1)\frac{\delta}{2}+\frac{1}{2}),\mu)$ and $\{(v_i,e_i)\}_{i=2}^{\mu-1}=\{( v'_j,j)\}_{j=2}^{\mu-1}$, there exists an $R$-graph with $\Omega(n^{\frac{\delta}{2}+\frac{1}{2(\mu-1)}}{(\log{n})}^{\frac{1}{2(\mu-1)}})$ edges. This proves the second part of (\ref{shg01}).

When $\delta$, $\mu>3$ are both odd. Assume that $\mu=2u+1$ for some $u\geq 2$. Then, we have
\begin{equation*}
(\mu-2)\frac{\delta}{2}+\frac{1}{2}=(\mu-2)\frac{\delta-1}{2}+u.
\end{equation*}
Thus, we also have $\gcd(\mu-2,(\mu-2)\frac{\delta}{2}+\frac{1}{2})=1$. By applying Lemma \ref{lem45} with $s=\mu-1$, $(v_1,e_1)=((\mu-1)R-((\mu-2)\frac{\delta}{2}+\frac{1}{2}),\mu-1)$ and $\{(v_i,e_i)\}_{i=2}^{\mu-1}=\{( v'_j,j)\}_{j=2,j\neq \mu-1}^{\mu}$, there exists an $R$-graph with $\Omega(n^{\frac{\delta}{2}+\frac{1}{2(\mu-2)}}{(\log{n})}^{\frac{1}{2(\mu-2)}})$ edges. This proves the third part of (\ref{shg01}).

Now, we turn to the proof of the rest part of (\ref{shg01}). When $\delta$ is odd and $\mu=3$, the conditions of Lemma \ref{lem45} no longer hold, thus we shall use the standard probabilistic method to prove the existence of such sparse hypergraph. Actually, we are going to prove the following stronger result.
\begin{claim}\label{claim2}
When both $\delta$ and $\mu$ are odd, there exists an $R$-uniform hypergraph $\mathcal{H}(V, E)$ that is simultaneously $\mathcal{G}_{R}(iR-\lfloor(i-1)\frac{\delta}{2}\rfloor-1,i)$-free for every $2\leq i\leq\mu$ with $|V|=n$ and $|E|=\Omega(n^{\frac{\delta}{2}+\frac{1}{2(\mu-2)}})$.
\end{claim}
\begin{proof}
Set $p:=p(n)=\varepsilon n^{\frac{\delta}{2}+\frac{1}{2(\mu-2)}-R}$ where $\varepsilon=\varepsilon(R,\delta,\mu)>0$ is a small constant to be determined. Construct an $R$-graph $\mathcal{H}_0\subseteq {V\choose R}$ randomly by choosing each member of ${V\choose R}$ independently with probability $p$. Let $X$ denote the number of edges in $\mathcal{H}_0$. Clearly, for $n$ sufficiently large,
\begin{equation*}
E[X]=p{n\choose R}\geq \frac{\varepsilon n^{\frac{\delta}{2}+\frac{1}{2(\mu-2)}}}{2R!}.
\end{equation*}
For $2\leq i\leq \mu$, let $\mathcal{Y}_i$ be the collection of all $i$ distinct edges of $\mathcal{H}_0$ whose union
contains at most $iR-\lfloor(i-1)\frac{\delta}{2}\rfloor-1$ vertices. Denote $Y_i$ as the size of $\mathcal{Y}_i$. Then,
\begin{align*}
E[Y_i]&\leq p^i{n\choose iR-\lfloor(i-1)\frac{\delta}{2}\rfloor-1}{iR-\lfloor(i-1)\frac{\delta}{2}\rfloor-1\choose R}^{i}\\
&\leq \varepsilon^{i} {iR-\lfloor(i-1)\frac{\delta}{2}\rfloor-1\choose R}^{i}n^{(\frac{\delta}{2}+\frac{1}{2(\mu-2)})i-\lfloor(i-1)\frac{\delta}{2}\rfloor-1}.
\end{align*}
Take $\varepsilon={(\mu R})^{-(3R)}$, since $\frac{\delta}{2}+\frac{1}{2(\mu-2)}\leq \frac{\lfloor(i-1)\frac{\delta}{2}\rfloor+1}{i-1}$ and ${iR-\lfloor(i-1)\frac{\delta}{2}\rfloor-1\choose R}^{i}\leq {\mu R\choose R}^{i}$, thus we have
\begin{align}\label{lb_ineq01}
E[Y_i]&\leq\varepsilon^{i} {iR-\lfloor(i-1)\frac{\delta}{2}\rfloor-1\choose R}^{i}n^{(\frac{\delta}{2}+\frac{1}{2(\mu-2)})i-\lfloor(i-1)\frac{\delta}{2}\rfloor-1} \nonumber\\
&<\frac{\varepsilon n^{\frac{\delta}{2}+\frac{1}{2(\mu-2)}}}{{\mu^{3i-3}(R!)^i}}\leq \frac{E[X]}{\mu^2R!}
\end{align}
for every $2\leq i\leq \mu$.

Applying Chernoff's inequality (see Corollary A.1.14 in \cite{TPM}) for $X$ and Markov's inequality for $Y_i$, it is easy to see that for each $2\leq i\leq\mu$ and sufficiently large $n$, we have
\begin{equation*}
Pr[X<0.9E[X]]<\frac{1}{2\mu} \text{ and } Pr[Y_i>2\mu E[Y_i]]<\frac{1}{2\mu}.
\end{equation*}
Therefore, with positive probability, there exists an $R$-graph $\mathcal{H}_0\subseteq{V\choose R}$ such that
\begin{equation*}
X\geq 0.9E[X] \text{ and } Y_i\leq 2\mu E[Y_i] \text{ for
each } 2\leq i\leq\mu.
\end{equation*}
Fix such $\mathcal{H}_0$, we construct a subgraph $\mathcal{H}_1$ from $\mathcal{H}_0$ by removing one edge from each member of $\mathcal{Y}_i$ in $\mathcal{H}_0$ for every $2\leq i\leq \mu$. By (\ref{lb_ineq01}), $\mathcal{H}_1$ satisfies $|E(\mathcal{H}_1)|=\Omega(n^{\frac{\delta}{2}+\frac{1}{2(\mu-2)}})$ and for each $2\leq i\leq \mu$, the union of any $i$ distinct edges in $\mathcal{H}_1$ contains at least $iR-\lfloor(i-1)\frac{\delta}{2}\rfloor$ vertices. Therefore, $\mathcal{H}_1$ is the desired $R$-graph and this proves the claim.
\end{proof}
Take $\mu=3$ in Claim \ref{claim2}, we have the fourth part of (\ref{shg01}). This completes the proof of Theorem \ref{main51}.
\end{proof}

\subsection{Optimal locally repairable codes with super-linear length based on sparse hypergraphs}

In this subsection, we are going to achieve our code constructions with the help of sparse hypergraphs. For LRCs with all symbol $(r,\delta)$-locality, we have the following result.

\begin{thm}\label{thm51}
For positive integers $\delta \geq 2$, $r\geq d-\delta$ and $d\geq 2\delta+1$. Let $R=r+\delta-1$, $\mu=\lfloor\frac{d-1}{\delta}\rfloor$ and $\mathcal{H}(V, E)$ be an $R$-uniform hypergraph with $V=\mathbb{F}_q$ that is simultaneously $\mathcal{G}_{R}(iR-\lfloor(i-1)\frac{\delta}{2}\rfloor-1,i)$-free for every $2\leq i\leq\mu$. Then, there exists an optimal $[n,k,d;(r,\delta)_a]_q$-LRC with length $n=R|E|$.
\end{thm}

\begin{proof}
Let $m=|E|$ and for each $e_i\in E(\mathcal{H})$, take $e_i$ as the generating set of Vandermonde matrices $\left(\begin{array}{c}
\mathbf{U}_i\\
\mathbf{V}_i
\end{array}\right)$ in Construction A. Note that $V(\mathcal{H})=\mathbb{F}_q$ and $\mathcal{H}$ is simultaneously $\mathcal{G}_{R}(iR-\lfloor(i-1)\frac{\delta}{2}\rfloor-1,i)$-free for $2\leq i\leq\mu$, therefore,
for any subset $S\subseteq[m]$ with $2\leq |S|\leq\lfloor\frac{d-1}{\delta}\rfloor$, we have
\begin{equation*}
|\bigcup_{i\in S}e_i|\geq R|S|-\lfloor(|S|-1)\frac{\delta}{2}\rfloor\geq(r+\frac{\delta}{2}-1)|S|+\frac{\delta}{2}.
\end{equation*}
Thus, the conclusion easily follows from Theorem \ref{main0}.
\end{proof}

As for LRCs with information $(r,\delta)$-locality, we have a similar result.

\begin{thm}\label{thm52}
For integers $r\geq 1$, $1\leq v\leq r$, $\delta \geq 2$ and $h\geq 0$. Let $R=r+\delta-1$, $\mu=\lfloor\frac{h+\delta-1}{\delta}\rfloor$ and $\mathcal{H}(V, E)$ be an $R$-uniform hypergraph with $V=\mathbb{F}_q\setminus G_{l+2}$ for some $h$-subset $G_{l+2}\subseteq \mathbb{F}_q$ that is simultaneously $\mathcal{G}_{R}(iR-\lfloor(i-1)\frac{\delta}{2}\rfloor-1,i)$-free for every $2\leq i\leq\mu$. Then, there exists an optimal $[n,k,h+\delta;(r,\delta)_i]_q$-LRC with length $n=R|E|-r+v$.
\end{thm}

\begin{proof}
Similarly, let $l+1=|E|$. Take any $e_{l+1}\in E(\mathcal{H})$, choose a $(v+\delta-1)$-subset of $e_{l+1}$ as the generating set of matrices $\mathbf{U}_{l+1}$ and $\mathbf{V}_{l+1}$, and for the rest $e_i\in E(\mathcal{H})$, take $e_i$ as the generating set of matrices $\mathbf{U}_i$ and $\mathbf{V}_i$ ($1\leq i\leq l$) in Construction B. Note that $\mathcal{H}$ is simultaneously $\mathcal{G}_{R}(iR-\lfloor(i-1)\frac{\delta}{2}\rfloor-1,i)$-free for $2\leq i\leq\mu$, therefore,
for any subset $S\subseteq[m]$ with $2\leq |S|\leq\lfloor\frac{h+\delta-1}{\delta}\rfloor$, we have
\begin{equation*}
|\bigcup_{i\in S}e_i|\geq
\begin{cases}
(r+\frac{\delta}{2}-1)|S|+\frac{\delta}{2}, \text{ when } l+1\notin S;\\
(r+\frac{\delta}{2}-1)|S|+\frac{\delta}{2}+v-r, \text{ when } l+1\in S.\\
\end{cases}
\end{equation*}
Thus, the conclusion easily follows from Corollary \ref{cor41}.
\end{proof}

Recall that in Theorem \ref{main0}, $r\geq d-\delta$ and $R|n$. When $2\delta+1\leq d\leq 3\delta$, one can get optimal LRCs with length $\Omega(q^{\delta})$ and minimum distance $d$ via packings or Steiner systems as in \cite{Jin19} and \cite{CMST20}. For $3\delta+1\leq d\leq 4\delta$, we have the following explicit construction.
\begin{cor}\label{cor42}
For $3\delta+1\leq d\leq4\delta$ and $r\geq d-\delta$, there exist explicit constructions of optimal $[n,k,d;(r,\delta)_a]_q$-LRCs with length
\begin{equation*}
	n = \begin{cases}
    Rq^{\frac{\delta}{2}+1-o(1)}, & \text{if}\  \delta \text{ is even};\\
    Rq^{\frac{\delta+1}{2}-o(1)}, & \text{if}\  \delta \text{ is odd}.\\
		\end{cases}
\end{equation*}
\end{cor}
\begin{proof}
When $\delta$ is even, take $l=\frac{\delta}{2}+1$ in Lemma \ref{lem42}, there exists an $R$-graph $\mathcal{H}_0$ which is both $\mathcal{G}_R(3R-\delta-1,3)$-free and $\mathcal{G}_R(2R-\frac{\delta}{2}-1,2)$-free with
\begin{equation*}
q^{l-o(1)}=q^{\frac{\delta}{2}+1-o(1)}
\end{equation*}
edges.

When $\delta$ is odd, take $l=\frac{\delta+1}{2}$ in Lemma \ref{lem42}, there exists an $R$-graph $\mathcal{H}_0$ which is both $\mathcal{G}_R(3R-\delta,3)$-free and $\mathcal{G}_R(2R-\frac{\delta+1}{2},2)$-free with
\begin{equation*}
q^{l-o(1)}=q^{\frac{\delta+1}{2}-o(1)}
\end{equation*}
edges.

Therefore, the conclusion follows from Theorem \ref{thm51}.
\end{proof}

\begin{cor}\label{cor43}
For $2\delta+1\leq h\leq3\delta$, there exist explicit constructions of optimal $[n,k,h+\delta;(r,\delta)_i]$-LRCs with length
\begin{equation*}
	n = \begin{cases}
    Rq^{\frac{\delta}{2}+1-o(1)}, & \text{if}\  \delta \text{ is even};\\
    Rq^{\frac{\delta+1}{2}-o(1)}, & \emph{if}\  \delta \text{ is odd}.\\
		\end{cases}
\end{equation*}
\end{cor}

\begin{proof}
Based on the sparse hypergraph given by Lemma \ref{lem42} and Construction B, the conclusion easily follows from Theorem \ref{thm52}.
\end{proof}

For LRCs with larger minimal distance, we have the following results from Theorem \ref{main51}, Theorem \ref{thm51} and Theorem \ref{thm52}.

\begin{cor}\label{lowerboundA}
For $\delta \geq 2$, $r\geq d-\delta$, $d\geq 3\delta+1$ and $q$ large enough. Let $\mu=\lfloor\frac{d-1}{\delta}\rfloor$, then there exists an optimal $[n,k,d;(r,\delta)_a]_q$-LRC of length
\begin{equation*}\label{lb01A}
n=\begin{cases}
\Omega(q^{\frac{\delta}{2}}{(q\log{q})}^{\frac{1}{\mu-1}}), ~\text{when}~\delta~\text{is even};\\
\Omega(q^{\frac{\delta}{2}}{(q\log{q})}^{\frac{1}{2(\mu-1)}}), ~\text{when}~\delta~\text{is odd}~{and}~\mu~\text{is even};\\
\Omega(q^{\frac{\delta}{2}}{(q\log{q})}^{\frac{1}{2(\mu-2)}}), ~\text{when both}~\delta~\text{and}~\mu>3~\text{are odd};\\
\Omega(q^{\frac{\delta+1}{2}}), ~\text{when}~\delta~\text{is odd}~{and}~\mu=3.
\end{cases}
\end{equation*}
\end{cor}

\begin{cor}\label{lowerboundB}
For $1\leq v\leq r$, $\delta \geq 2$ and $h\geq 2\delta+1$ and $q$ large enough. Let $\mu=\lfloor\frac{h+\delta-1}{\delta}\rfloor$, then there exists an optimal $[n,k,h+\delta;(r,\delta)_i]_q$-LRC of length
\begin{equation*}\label{lb01B}
n=\begin{cases}
\Omega(q^{\frac{\delta}{2}}{(q\log{q})}^{\frac{1}{\mu-1}}), ~\text{when}~\delta~\text{is even};\\
\Omega(q^{\frac{\delta}{2}}{(q\log{q})}^{\frac{1}{2(\mu-1)}}), ~\text{when}~\delta~\text{is odd}~{and}~\mu~\text{is even};\\
\Omega(q^{\frac{\delta}{2}}{(q\log{q})}^{\frac{1}{2(\mu-2)}}), ~\text{when both}~\delta~\text{and}~\mu>3~\text{are odd};\\
\Omega(q^{\frac{\delta+1}{2}}), ~\text{when}~\delta~\text{is odd}~{and}~\mu=3.
\end{cases}
\end{equation*}
\end{cor}

\begin{rem}
For more details about sparse hypergraphs and other related applications, we recommend \cite{bujtas2015turan} and \cite{ST20} for interested readers.
\end{rem}

In Table I and Table II, we have listed all the known parameters of optimal LRCs of super-linear length together with our results. As one can see, for optimal $(r,\delta)_a$-LRCs:
\begin{itemize}
  \item when $\delta=2$, our results from Corollary \ref{cor42} and Corollary \ref{lowerboundA} agree with those in \cite{XY19} for $d=7,8$ and $d\geq 11$; for $d=9,10$, Xing and Yuan \cite{XY19} provided longer codes;
  \item when $\delta>2$ and $2\delta+1\leq d\leq 3\delta$, Cai et.al \cite{CMST20} provided the longest known codes of length $\Omega(q^{\delta})$ which meets the upper bound for the case $d=2\delta+1$;
  \item when $\delta>2$ and $d\geq 3\delta+1$, Corollary \ref{cor42} gives the longest known codes for $d\leq 4\delta$ and $\delta$ is even; Corollary \ref{lowerboundA} gives the longest known codes for other cases.
\end{itemize}
For optimal $(r,\delta)_i$-LRCs:
\begin{itemize}
  \item when $\delta>2$ and $\delta+1\leq d\leq 2\delta$, Corollary \ref{cor40} provides a code of arbitrarily long length;
  \item when $\delta>2$ and $2\delta+1\leq d\leq 3\delta$, Cai and Schwartz \cite{CS20} provided codes of order optimal length $\Omega(q^{\delta})$;
  \item when $\delta>2$ and $d\geq 3\delta+1$, Corollary \ref{cor43} gives the longest known codes for $d\leq 4\delta$ and $\delta$ is even; Corollary \ref{lowerboundB} gives the longest known codes for other cases.
\end{itemize}

\begin{table}
  \caption{Optimal $(r,\delta)_a$-LRCs over $\mathbb{F}_q$ with super-linear lengths and corresponding upper bounds}
  \begin{center}
  \begin{tabular}{cccc}
  \toprule
  Distance & Other conditions & Length~ & ~Upper Bound \\
  \midrule
  $d=5,6$ & $\delta=2$, $r\geq d-2$, $r+1|n$ & $\Omega(q^{2})$ (\cite{Jin19}, \cite{guruswami2019long} and \cite{CMST20}) & $\begin{cases}O(q^2), d=5\\O(q^3), d=6\end{cases}$(\cite{guruswami2019long})\\
  \midrule
  $d=7,8$ & $\delta=2$, $r\geq d-2$, $r+1|n$ & $\Omega(q^{2-o(1)})$ (\cite{XY19}) & $\begin{cases}O(q^3), d=7\\O(q^4), d=8\end{cases}$(\cite{guruswami2019long})\\
  \midrule
  $d=9,10$ & $\delta=2$, $r\geq d-2$, $r+1|n$ & $\Omega(q^{\frac{3}{2}-o(1)})$ (\cite{XY19}) & $\begin{cases}O(q^\frac{5}{2}), d=9\\O(q^3), d=10\end{cases}$(\cite{guruswami2019long})\\
  \midrule
  $d\geq 11$ & $\delta=2$, $r\geq d-2$, $r+1|n$ & \tabincell{c}{$\Omega(q(q\log{q})^{\frac{1}{\lfloor(d-3)/2\rfloor}})$ \\(\cite{XY19} and \cite{ST20})}& $\begin{cases}O(dq^3), 4\nmid d\\O(dq^{3+\frac{4}{d-4}}), 4|d\end{cases}$(\cite{guruswami2019long})\\
  \midrule
  $\delta+1\leq d\leq 2\delta$ & \tabincell{c}{$d\leq r+\delta-1\leq q$, \\$r+\delta-1|n$} & $\infty$ (\cite{CMST20} and \cite{ZL20}) & $\infty$ \\
  \midrule
  $d\ge 2\delta+1$ & $r\geq d-\delta+1$, $r+\delta-1|n$ & $\Omega(q^{\frac{\delta}{\lceil d/\delta\rceil-2}})$ (\cite{CMST20}) & \tabincell{c}{$\begin{cases}O(q^{\frac{2(d-\delta-1)}{\lfloor(d-1)/\delta\rfloor}-1}),\lfloor\frac{d-1}{\delta}\rfloor~\text{odd}\\O(q^{\frac{2(d-\delta)}{\lfloor(d-1)/\delta\rfloor}-1}),\lfloor\frac{d-1}{\delta}\rfloor~\text{even}\end{cases}$ (\cite{CMST20})}\\
  \midrule
  $d\geq 2\delta+1$ & $r\geq d-\delta+1$, $r+\delta-1|n$ & $\Omega(q^{1+\lfloor\frac{\delta^2}{d-\delta}\rfloor})$~(\cite{CMST20}) & ------\\
  \midrule
  $3\delta+1\leq d\leq 4\delta$ & \tabincell{c}{$\delta$ even, $r\geq d-\delta$,\\ $r+\delta-1|n$} & $\Omega(q^{1+\frac{\delta}{2}-o(1)})$~ (Corollary \ref{cor42})& ------\\
  \midrule
  $3\delta+1\leq d\leq 4\delta$ & \tabincell{c}{$\delta$ odd, $r\geq d-\delta$,\\ $r+\delta-1|n$} & $\Omega(q^{\frac{\delta+1}{2}})$~ (Corollary \ref{lowerboundA})& ------\\
  \midrule
  $d\geq 3\delta+1$ & \tabincell{c}{$\delta$ even, $r\geq d-\delta$,\\ $r+\delta-1|n$} & \tabincell{c}{$\Omega(q^{\frac{\delta}{2}}{(q\log{q})}^{\frac{1}{\lfloor\frac{d-1}{\delta}\rfloor-1}})$\\(Corollary \ref{lowerboundA})}& ------\\
  \midrule
  $d\geq 3\delta+1$ & \tabincell{c}{$\delta$ odd, $\lfloor\frac{d-1}{\delta}\rfloor$ even, \\$r\geq d-\delta$, $r+\delta-1|n$} & \tabincell{c}{$
  \Omega(q^{\frac{\delta}{2}}{(q\log{q})}^{\frac{1}{2(\lfloor\frac{d-1}{\delta}\rfloor-1)}})$ \\(Corollary \ref{lowerboundA})}& ------\\
  \midrule
  $d\geq 3\delta+1$ & \tabincell{c}{$\delta$ odd, $\lfloor\frac{d-1}{\delta}\rfloor>3$ odd, \\$r\geq d-\delta$, $r+\delta-1|n$} & \tabincell{c}{$
  \Omega(q^{\frac{\delta}{2}}{(q\log{q})}^{\frac{1}{2(\lfloor\frac{d-1}{\delta}\rfloor-2)}})$ \\(Corollary \ref{lowerboundA})}& ------\\
  \bottomrule
  \end{tabular}
  \end{center}
\end{table}

\begin{table}
  \caption{Optimal $(r,\delta)_i$-LRCs over $\mathbb{F}_q$ with super-linear lengths and corresponding upper bounds}
  \begin{center}
  \begin{tabular}{cccc}
  \toprule
  Distance & Other conditions & Length~ & ~Upper Bound \\
  \midrule
  $\delta+1\leq d\leq 2\delta$ & \tabincell{c}{$q\geq r+\delta-1$} & $\infty$ (Corollary \ref{cor40}) & $\infty$ \\
  \midrule
  $d\ge 2\delta+1$ & Null & \tabincell{c}{$\Omega(q^{\tau+1})$, where \\$\tau=\max\{x\in \mathbb{N}^*:\lceil\frac{\delta}{x}\rceil=\lceil\frac{d-\delta}{\delta}\rceil\}$ \\(\cite{CS20})} & \tabincell{c}{$\begin{cases}O(q^{\frac{2(d-\delta-a-1)}{T(a)-1}-1}),T(a)~\text{odd}\\O(q^{\frac{2(d-\delta-a)}{T(a)-1}-1}),T(a)~\text{even}\end{cases}$ (\cite{CS20})\\when $r|k$ and $T(a)\geq 2$,\\where $T(a)=\lfloor\frac{d-a-1}{\delta}\rfloor$, for any\\$0\leq a\leq d-\delta$}\\
  \midrule
  $3\delta+1\leq d\leq 4\delta$ & \tabincell{c}{$\delta$ even} & $\Omega(q^{1+\frac{\delta}{2}-o(1)})$~ (Corollary \ref{cor43})& ------\\
  \midrule
  $3\delta+1\leq d\leq 4\delta$ & \tabincell{c}{$\delta$ odd} & \tabincell{c}{$\Omega(q^{\frac{\delta+1}{2}})$ (Corollary \ref{lowerboundB})}& ------\\
  \midrule
  $d\geq 3\delta+1$ & \tabincell{c}{$\delta$ even} & \tabincell{c}{$\Omega(q^{\frac{\delta}{2}}{(q\log{q})}^{\frac{1}{\lfloor\frac{d-1}{\delta}\rfloor-1}})$\\(Corollary \ref{lowerboundB})}& ------\\
  \midrule
  $d\geq 3\delta+1$ & \tabincell{c}{$\delta$ odd and $\lfloor\frac{d-1}{\delta}\rfloor$ even} & \tabincell{c}{$
  \Omega(q^{\frac{\delta}{2}}{(q\log{q})}^{\frac{1}{2(\lfloor\frac{d-1}{\delta}\rfloor-1)}})$ \\(Corollary \ref{lowerboundB})}& ------\\
  \midrule
  $d\geq 3\delta+1$ & \tabincell{c}{$\delta$ odd and $\lfloor\frac{d-1}{\delta}\rfloor>3$ odd} & \tabincell{c}{$
  \Omega(q^{\frac{\delta}{2}}{(q\log{q})}^{\frac{1}{2(\lfloor\frac{d-1}{\delta}\rfloor-2)}})$ \\(Corollary \ref{lowerboundB})}& ------\\
  \bottomrule
  \end{tabular}
  \end{center}
\end{table}

\section{Applications: Constructions of H-LRCs and generalized sector-disk codes}\label{sec6}

In this section, we present two applications of Constructions A and B, respectively. In Subsection \ref{sec61}, based on Construction A, we construct optimal H-LRCs with super-linear length, which improves the results given by \cite{Zhang20}. In Subsection \ref{sec62}, based on Construction B, we provide two constructions of generalized sector-disk codes, which provide a code of unbounded length.

\subsection{Optimal Codes with Hierarchical Locality}\label{sec61}

The conception of \emph{hierarchical locality} was first introduced by Sasidharan et al. in \cite{SAK15}. The authors considered the intermediate situation when the code can correct a single erasure by contacting a small number of helper nodes, while at the same time maintains local recovery of multiple erasures. Such codes are called locally recoverable codes with hierarchical locality, its formal definition is given as follows.
\begin{definition}\cite{SAK15}\label{HLRC}
Let $2\leq \delta_2<\delta_1$ and $r_2\leq r_1$ be positive integers. An $[n,k,d]_q$ code $\mathcal{C}$ is an H-LRC with parameters $[(r_1,\delta_1),(r_2,\delta_2)]$ if for every $i\in [n]$, there is a punctured code $\mathcal{C}_i$ such that $i\in Supp(\mathcal{C}_i)$ and the following conditions hold
\begin{itemize}
  \item [1)] $\dim(\mathcal{C}_i)\leq r_1$;
  \item [2)] $d(\mathcal{C}_i)\geq \delta_1$;
  \item [3)] $\mathcal{C}_i$ is a code with $(r_2,\delta_2)$-locality.
\end{itemize}
\end{definition}
For each $i\in [n]$, the punctured code $\mathcal{C}_i$ associated with $c_i$ is referred to as a \emph{middle code} of $\mathcal{C}$. In \cite{SAK15}, the authors extended the Singleton-type bound (\ref{singleton_bound}) and proved the following bound for H-LRCs with parameters $[(r_1,\delta_1),(r_2,\delta_2)]$:
\begin{equation}\label{singleton_bound_H}
d\leq n-k+1-(\left\lceil\frac{k}{r_2}\right\rceil-1)(\delta_2-1)-(\left\lceil\frac{k}{r_1}\right\rceil-1)(\delta_1-\delta_2).
\end{equation}
H-LRCs with parameters $[(r_1,\delta_1),(r_2,\delta_2)]$ that attain this bound are called optimal. Using Reed-Solomon codes, Sasidharan et al. \cite{SAK15} construct optimal H-LRCs of length $n\leq q-1$. Later in \cite{BBV2019}, Ballentine et al. presented a general construction of H-LRCs via maps between algebraic curves. From the perspective of parity-check matrices, Zhang and Liu \cite{ZL20} obtained a family of optimal H-LRCs with parameters $[(r_1,\delta_1),(r_2,\delta_2)]$ and minimum distance $d\leq 3\delta_2$. Recently, in \cite{Zhang20}, Zhang extended the constructions in \cite{ZL20} and obtained H-LRCs with new parameters.

Based on Construction A and a correspondence between optimal $(r_2,\delta_2)_a$-LRCs and optimal H-LRCs with parameters $[(r_1,\delta_1),(r_2,\delta_2)]$ from \cite{ZL20}, we have the following result.

\begin{thm}\label{opt_HLRCs}
For positive integers $m_2$, $r_2$, $\delta_2\geq 2$ and $d_2< r_2+\delta_2$. Let $\mathbf{H}_{middle}=\mathbf{H}_{middle}(m_2,r_2,\delta_2,d_2)$ be the parity-check matrix in (\ref{pcm1}) from Construction A with parameters $m=m_2$, $r=r_2$, $\delta=\delta_2$ and $d=d_2$. For positive integer $m_1<\frac{r_2}{d_2-\delta_2}$, define
\begin{equation}\label{pcm_HLRC}
\mathbf{H}(m_1,\mathbf{H}_{middle})=\left(\begin{array}{cccc}
\mathbf{H}_{middle} &O &\cdots &O\\
O &\mathbf{H}_{middle} &\cdots &O\\
\vdots &\vdots &\ddots &\vdots\\
O &O &\cdots &\mathbf{H}_{middle}
\end{array}\right),
\end{equation}
where there are $m_1$ $\mathbf{H}_{middle}$s on the diagonal. Let $r_1$ and $\delta_1$ be positive integers satisfying
\begin{equation}\label{para_cond}
r_1(1-\frac{1}{m_1})<m_2r_2-d+\delta_2\leq r_1, \text{ and } \delta_1=d_2.
\end{equation}
If for any subset $S\subseteq[m_2]$ with $2\leq |S|\leq\lfloor\frac{d_2-1}{\delta_2}\rfloor$, we have
$|\bigcup_{i\in S}G_i|\geq (r_2+\frac{\delta_2}{2}-1)|S|+\frac{\delta_2}{2}$. Then, the code $\mathcal{C}$ with parity-check matrix $\mathbf{H}(m_1,\mathbf{H}_{middle})$ is an optimal $[n,k,d]_q$ H-LRC with parameters $[(r_1,\delta_1),(r_2,\delta_2)]$, where $n=m_1m_2(r_2+\delta_2-1)$, $k=m_1(m_2r_2-d_2+\delta_2)$ and $d=d_2$.
\end{thm}

\begin{proof}
The proof is a routine check of the conditions in Definition \ref{HLRC} and the equality in (\ref{singleton_bound_H}).

Clearly, $n=m_1m_2(r_2+\delta_2-1)$ and $rank(\mathbf{H}(m_1,\mathbf{H}_{middle}))=m_1rank(\mathbf{H}_{middle})$. By Theorem \ref{main0}, the code $\mathcal{C}_{middle}$ with parity-check matrix $\mathbf{H}_{middle}$ is an optimal $[m_2(r_2+\delta_2-1),m_2r_2-d_2+\delta_2,d_2]_q$-LRC with $(r_2,\delta_2)_a$-locality. This verifies condition $3)$ in Definition \ref{def01}. And conditions $1),2)$ in Definition \ref{HLRC} follow from (\ref{para_cond}). Moreover, we also have $rank(\mathbf{H}_{middle})=m_2(\delta_2-1)+d_2-\delta_2$, which leads to $k=m_1(m_2r_2-d_2+\delta_2)$.

It remains to verify the optimality of $\mathcal{C}$ w.r.t. bound (\ref{singleton_bound_H}). From $d_2< r_2+\delta_2$, $m_1<\frac{r_2}{d_2-\delta_2}$ and $r_1(1-\frac{1}{m_1})<m_2r_2-d+\delta_2$, we have $\left\lceil\frac{k}{r_2}\right\rceil=m_1m_2$ and $\left\lceil\frac{k}{r_1}\right\rceil=m_1$. Therefore,
\begin{align*}
&n-k+1-(\left\lceil\frac{k}{r_2}\right\rceil-1)(\delta_2-1)-(\left\lceil\frac{k}{r_1}\right\rceil-1)(\delta_1-\delta_2)\\
=&m_1(m_2(\delta_2-1)+d_2-\delta_2)+1-(m_1m_2-1)(\delta_2-1)-(m_1-1)(d_2-\delta_2)\\
=&m_1(d_2-\delta_2)+\delta_2-(m_1-1)(d_2-\delta_2)=d_2.
\end{align*}

This completes the proof of Theorem \ref{opt_HLRCs}.
\end{proof}

Analogous to the case for $(r,\delta)_a$-LRCs, as immediate consequences of Theorem \ref{opt_HLRCs}, we have the following corollaries.

\begin{cor}\label{lowerbound_HLRC0}
Let $r_1$, $r_2$, $\delta_1$, $\delta_2$, $3\delta_2+1\leq d_2\leq4\delta_2$ be those parameters defined in Theorem \ref{opt_HLRCs}. Then, there exist explicit constructions of optimal $[n,k,d_2]_q$ H-LRCs with parameters $[(r_1,\delta_1),(r_2,\delta_2)]$ of length
\begin{equation*}
	n = \begin{cases}
    (r_2+\delta_2-1)q^{\frac{\delta_2}{2}+1-o(1)}, & \text{if}\  \delta_2 \text{ is even};\\
    (r_2+\delta_2-1)q^{\frac{\delta_2+1}{2}-o(1)}, & \text{if}\  \delta_2 \text{ is odd}.\\
		\end{cases}
\end{equation*}
\end{cor}

\begin{cor}\label{lowerbound_HLRC}
Let $r_1$, $r_2$, $\delta_1$, $d_2\geq 3\delta_2+1$ be those parameters defined in Theorem \ref{opt_HLRCs} and $\mu=\lfloor\frac{d_2-1}{\delta_2}\rfloor$. For $q$ sufficiently large, there exists an optimal $[n,k,d_2]_q$ H-LRC with parameters $[(r_1,\delta_1),(r_2,\delta_2)]$ of length
\begin{equation}\label{lb01_HLRC}
n=\begin{cases}
\Omega(q^{\frac{\delta_2}{2}}{(q\log{q})}^{\frac{1}{\mu-1}}), ~\text{when}~\delta_2~\text{is even};\\
\Omega(q^{\frac{\delta_2}{2}}{(q\log{q})}^{\frac{1}{2(\mu-1)}}), ~\text{when}~\delta_2~\text{is odd}~{and}~\mu~\text{is even};\\
\Omega(q^{\frac{\delta_2}{2}}{(q\log{q})}^{\frac{1}{2(\mu-2)}}), ~\text{when both}~\delta_2~\text{and}~\mu>3~\text{are odd};\\
\Omega(q^{\frac{\delta_2+1}{2}}), ~\text{when}~\delta_2~\text{is odd}~{and}~\mu=3.
\end{cases}
\end{equation}
\end{cor}

\begin{rem}\label{rem4}
\begin{itemize}
  \item[(i)] As mentioned in Remark \ref{rem1}, the construction of optimal $(r,\delta)_a$-LRCs in \cite{Zhang20} is under the condition (\ref{cond_zhang}). As a consequence, H-LRCs generated from this construction require the generating sets of the middle code satisfying (\ref{cond_zhang}). Therefore, compared to the construction in \cite{Zhang20}, Theorem \ref{opt_HLRCs} provides a way to construct H-LRCs under a more relaxed condition.
  \item[(ii)] When $\delta_2+1\leq d_2\leq 2\delta_2$, like Theorem \ref{thm31}, Zhang and Liu \cite{ZL20} provide a construction of optimal $[n,k,d]_q$ H-LRCs with parameters $[(r_1,\delta_1),(r_2,\delta_2)]$ with unbounded length. When $2\delta_2+1\leq d_2\leq 3\delta_2$, H-LRCs obtained from Theorems IV.3 and IV.4 in \cite{ZL20} can have length $\Omega(q^2)$ through $(q,r_2+\delta_2-1,1)$-packings. When $d_2\geq 3\delta_2+1$, Corollary \ref{lowerbound_HLRC0} and Corollary \ref{lowerbound_HLRC} give the longest known optimal H-LRCs for these cases.
\end{itemize}
\end{rem}

\subsection{Generalized Sector-Disk Codes}\label{sec62}

Aiming to construct codes that can recover erasure patterns beyond the minimum distance, Cai and Schwartz \cite{CS20} relaxed the restrictions of sector-disk codes \cite{PB14} and considered the following array codes.

\begin{definition}\cite{CS20}\label{GSD}
Let $\mathcal{C}$ be an optimal $[n,k,d;(r,\delta)_i]_q$-LRC. Then the code $\mathcal{C}$ is said to be a $(\gamma,s)$-generalized sector disk code (GSD code) if the codewords can be arranged into an array
\begin{equation*}
C=\left(\begin{array}{cccc}
c_{1,1} &c_{1,2} &\cdots &c_{1,a}\\
c_{2,1} &c_{2,2} &\cdots &c_{2,a}\\
\vdots &\vdots &\ddots &\vdots\\
c_{b,1} &c_{b,2} &\cdots &c_{b,a}\\
\end{array}\right)
\end{equation*}
such that:
\begin{description}
  \item[(i)] all the erasure patterns consisting of any $\gamma$ columns and additional $s$ sectors can be recovered;
  \item[(ii)] $\gamma b+s>d-1$.
\end{description}
\end{definition}
In \cite{CS20}, based on locally repairable codes with information locality constructed from regular packings, Cai and Schwartz obtained GSD codes with super-linear length for several different $(\gamma,s)$s. As an application of Theorem \ref{thm41}, we have the following construction.

\textbf{Construction C:} For positive integers $r$, $\delta\geq 2$ and $1\leq h\leq \delta$. Let $S$ be an $h$-subset of $\mathbb{F}_q$ and $G$ be an $(r+\delta-1)$-subset of $\mathbb{F}_q\setminus S$. For any positive integer $l\geq 1$, let $n=(l+1)(r+\delta-1)+h$ and $k=(l+1)r$. Then take $G_{l+2}=S$ and $G_i=G$ for $1\leq i\leq l+1$, we can obtain an $[n,k,h+\delta;(r,\delta)_i]$-LRC $\mathcal{C}_0$ by Construction B. Denote $G=\{x_1,\ldots,x_{r+\delta-1}\}$ and $c=(c_{1,1},\ldots,c_{1,r+\delta-1},\ldots,c_{l+1,1},\ldots,c_{l+1,r+\delta-1},c_{l+2,1},\ldots,c_{l+2,h})$ for any $c\in\mathcal{C}_0$. Define column vectors $\mathbf{v}_{x_{a}}\in \mathbb{F}_{q}^{l+2}$ for $a\in [r+\delta-1]$ as
\begin{equation*}
\mathbf{v}_{x_{a}}^{T}=
(c_{i_{x_a,1},j_{x_a,1}},c_{i_{x_a,2},j_{x_a,2}},\ldots,c_{i_{x_a,l+1},j_{x_a,l+1}},c'_{l+2,a}),
\end{equation*}
where the generating element corresponding to $c_{i_{x_a,b},j_{x_a,b}}$ satisfies $g_{i_{x_a,b},j_{x_{a},b}}=x_a$ for $1\leq b\leq l+1$, $c'_{l+2,a}=c_{l+2,a}$ for $1\leq a\leq h$ and $c'_{l+2,a}=0$ for $h+1\leq a\leq r+\delta-1$.

\begin{thm}\label{GSD0}
Let $\mathcal{C}$ be the $(l+2)\times (r+\delta-1)$ array code generated by Construction C. Then,
\begin{itemize}
  \item when $\gamma\leq h$, the code $\mathcal{C}$ is a $(\gamma, h+\delta-1-2\gamma)$-GSD code;
  \item when $h<\gamma< \delta-1$, the code $\mathcal{C}$ is a $(\gamma, \delta-1-\gamma)$-GSD code.
\end{itemize}
\end{thm}
\begin{proof}
According to Definition \ref{GSD} and Theorem \ref{thm41}, we only need to show that erasure patterns consisting of $\gamma$ columns and any other $s$ erasures satisfy (\ref{ineq_erasurepattern}) and (\ref{ineq_evaluationpoint}), where
\begin{equation*}
s=\begin{cases}
h+\delta-1-2\gamma,~\text{when}~\gamma\leq h;\\
\delta-1-\gamma,~\text{when}~h<\gamma<\delta-1.
\end{cases}
\end{equation*}

Let $\mathcal{F}=\{F_1,\ldots,F_{l+2}\}$ be the erasure pattern formed by given $\gamma$ columns and other $s$ erasures. Denote $\mathcal{F}'=\{F_{i}\in \mathcal{F}: |F_{i}|\geq \delta\}\cup\{F_{l+2}\}$. Clearly, we have
\begin{align*}
|\bigcup_{F_{i}\in \mathcal{F}'}F_{i}|+|F_{l+2}|\leq |\bigcup_{F_{i}\in \mathcal{F}}F_{i}|+|F_{l+2}|&\leq
\begin{cases}
2\gamma+s,~\text{when}~\gamma\leq h;\\
\gamma+h+s,~\text{when}~h<\gamma<\delta-1
\end{cases}\\
&\leq h+\delta-1.
\end{align*}
Therefore, $\mathcal{F}$ satisfies (\ref{ineq_erasurepattern}). Moreover, since $G_i=G$ for all $i\in [l+1]$, thus (\ref{ineq_evaluationpoint}) holds naturally.
\end{proof}

From Corollary \ref{cor40}, the code $\mathcal{C}$ generated by Construction C can be arbitrarily long and its minimal distance $d=h+\delta$ satisfies $\delta+1\leq d\leq 2\delta$. For general $d\geq 2\delta+1$, based on Theorem \ref{main1}, we can extend Cai and Schwartz's construction as follows.

\textbf{Construction D:} Given positive integers $r$, $1\leq v< r$ and $\delta\geq 2$. Let $S$ be an $(r-v)$-subset of $\mathbb{F}_q$ and $\mathcal{H}(V,E)$ be a $t$-regular $(t\geq 2)$ $R$-uniform hypergraph with $V=\mathbb{F}_q\setminus S$ that is $\mathcal{G}_{R}(iR-\lfloor(i-1)\frac{\delta}{2}\rfloor-1,i)$-free for every $2\leq i\leq {r-v+\delta-1\choose \delta}$. Let $E=\{e_{i}\}_{i=1}^{|E|}$, $G_i=e_i$ for $1\leq i\leq |E|-1$ and $G_{|E|}$ be a $(v+\delta-1)$-subset of $e_{|E|}$. Let $n=t(q-r+v)$ and $k=(|E|-1)r+v$. Based on $S$ and $\mathcal{H}$, we can obtain an $[n,k,r-v+\delta;(r,\delta)_i]$-LRC $\mathcal{C}_0$ by Construction B. Denote $e_{|E|}\setminus G_{|E|}=\{x_1,\ldots,x_{r-v}\}$, $V=\{x_1,\ldots,x_{q-r+v}\}$ and $c=(c_{1,1},\ldots,c_{1,r+\delta-1},\ldots,c_{|E|,1},\ldots,c_{|E|,v+\delta-1},c_{|E|+1,1},\ldots,c_{|E|+1,r-v})$ for any $c\in\mathcal{C}_0$. Define column vectors $\mathbf{v}_{x_{a}}\in \mathbb{F}_{q}^{t}$ for $a\in [q-r+v]$ as
\begin{equation*}
\mathbf{v}_{x_{a}}^{T}=
\begin{cases}
(c_{i_{x_a,1},j_{x_a,1}},c_{i_{x_a,2},j_{x_a,2}},\ldots,c_{i_{x_a,t-1},j_{x_a,t-1}},c_{|E|+1,a}),~\text{if}~1\leq a\leq r-v;\\
(c_{i_{x_a,1},j_{x_a,1}},c_{i_{x_a,2},j_{x_a,2}},\ldots,c_{i_{x_a,t},j_{x_a,t}}),~\text{otherwise},
\end{cases}
\end{equation*}
where the generating element corresponding to $c_{i_{x_a,b},j_{x_a,b}}$ satisfies $g_{i_{x_a,b},j_{x_{a},b}}=x_a$, $1\leq b\leq t-1$ for $1\leq a\leq r-v$ and $1\leq b\leq t$ for $r-v+1\leq a\leq q-r+v$.

%

\begin{thm}\label{GSD_SP}
Let $\mathcal{C}$ be the $t\times (q-r+v)$ array code generated by Construction D. Then,
\begin{itemize}
  \item when $\gamma\leq r-v$, the code $\mathcal{C}$ is a $(\gamma, r-v+\delta-1-2\gamma)$-GSD code;
  \item when $r-v<\gamma< \delta-1$, the code $\mathcal{C}$ is a $(\gamma, \delta-1-\gamma)$-GSD code.
\end{itemize}
\end{thm}
\begin{proof}
According to Definition \ref{GSD} and Theorem \ref{main1}, we only need to show that erasure patterns consisting of $\gamma$ columns and any other $s$ erasures satisfy (\ref{ineq_erasurepattern}) and (\ref{ineq_evaluationpoint}), where
\begin{equation*}
s=\begin{cases}
r-v+\delta-1-2\gamma,~\text{when}~\gamma\leq r-v;\\
\delta-1-\gamma,~\text{when}~r-v<\gamma< \delta-1.
\end{cases}
\end{equation*}

Let $\mathcal{F}=\{F_1,\ldots,F_{|E|+1}\}$ be the erasure pattern formed by $\gamma$ given columns and other $s$ erasures. Denote $\mathcal{F}'=\{F_{i}\in \mathcal{F}: |F_{i}|\geq \delta\}\cup\{F_{|E|+1}\}$ and $|\bigcup_{F_{i}\in \mathcal{F}'}F_{i}|+|F_{|E|+1}|\leq r-v+\delta-1$ follows from the choice of $s$. Therefore, $\mathcal{F}$ satisfies (\ref{ineq_erasurepattern}).

On the other hand, denote $I_{\mathcal{F}'}=\{i\in [|E|]: F_i\in \mathcal{F}'\}$. Note that $|\bigcup_{i\in I_{\mathcal{F}'}}F_{i}|\leq r-v+\delta-1$ and $|F_{i}|\geq \delta$ for each $i\in I_{\mathcal{F}'}$, which indicates that
\begin{equation*}
|I_{\mathcal{F}'}|\leq {r-v+\delta-1\choose \delta}.
\end{equation*}
Therefore, (\ref{ineq_evaluationpoint}) follows from the sparsity of $\mathcal{H}$.
\end{proof}

\begin{rem}\label{rem3}
Unfortunately, all known results about large sparse hypergraphs can not guarantee the regularity of every vertex of $\mathcal{H}$. A standard probabilistic argument like Claim \ref{claim2} can only provide sparse hypergraphs with bounded degree. Thus, more advanced methods are required to construct large regular sparse hypergraphs.
\end{rem}

\section{Conclusions}\label{sec7}

In this paper, we provide general constructions for both optimal $(r,\delta)_a$-LRCs and optimal $(r,\delta)_i$-LRCs. Based on a connection between sparse hypergraphs and optimal $(r,\delta)$-LRCs, we obtain optimal $(r,\delta)_a$-LRCs and optimal $(r,\delta)_i$-LRCs with super-linear (in $q$) length. This improves all known results when the minimal distance $d$ satisfies $d\geq 3\delta+1$. Moreover, as applications, we provide new constructions for H-LRCs and GSD codes.

As shown in Theorem \ref{main1}, codes generated by Construction B can recover special erasure patterns beyond the minimal distance, which enables us to further construct GSD codes. This phenomenon also appears in codes from Construction A. Note that the parity check matrix in (\ref{pcm1}) have similar structure as that for MR-LRC (see \cite{GGY20}), therefore, it's worth trying to obtain longer MR-LRCs using similar approaches.

According to Tables I and II, there are gaps between our constructions and upper bounds on the code length given in \cite{CMST20} and \cite{CS20}. Therefore, improvements of the upper bounds and constructions of longer codes will be interesting topics for future work. Moreover, to our knowledge, explicit constructions of large sparse hypergraphs are very rare. Results of Lemma \ref{lem45} and therefore results of Theorem \ref{thm51} are both from the perspective of probabilistic existence. Therefore, explicit constructions or algorithmic constructions in polynomial time (like Theorem 4.2 in \cite{XY19}) for optimal $(r,\delta)$-LRCs with super-linear length are also worth studying.

\bibliographystyle{abbrv}
\bibliography{New_Constructions_of_Optimal_Locally_Repairable_Codes_with_Super-Linear_Length}
\end{document}